\documentclass[10pt,usenames,dvipsnames]{article}
\usepackage{enumerate}
\usepackage{amsmath,amssymb}
\usepackage{caption}
\usepackage[usenames]{color}
\usepackage{ying}
\usepackage{multirow}
\usepackage{rotating}
\usepackage{fullpage}
\usepackage{authblk}
\usepackage{enumerate}
\usepackage[margin=1in,footskip=0.25in]{geometry}
\usepackage{comment}

\usepackage{mathtools}
\mathtoolsset{showonlyrefs}
\usepackage{xcolor}
\usepackage{mathtools}
\usepackage{pdfsync} 

\usepackage{graphicx,amssymb}
\usepackage{xcolor}
\usepackage[colorlinks,
            linkcolor=red,
            anchorcolor=blue,
            citecolor=blue
            ]{hyperref}
\usepackage{algorithm}
\usepackage{algorithmic}

\let\hat\widehat
\let\tilde\widetilde

\theoremstyle{plain}

\def\sep{\;\big|\;}

\usepackage{hyperref}
\def\##1\#{\begin{align}#1\end{align}}
\def\$#1\${\begin{align*}#1\end{align*}}

\usepackage{tikz,tabularx}
\usetikzlibrary{patterns}
\usetikzlibrary{arrows}
\usetikzlibrary{tikzmark, positioning, fit, shapes.misc}
\usetikzlibrary{decorations.pathreplacing, calc}
\tikzset{brace/.style={decorate, decoration={brace}},
  brace mirrored/.style={decorate, decoration={brace,mirror}},
}
\def\keywords{\vspace{.5em}
{\textit{Keywords}:\,\relax%
}}

\usepackage{xcolor}

\newtheorem{principle}{Principle}
\newtheorem{abc}{Rule}

\newcommand{\mcs}{\ensuremath{\mathcal{S}}}
\newcommand{\mck}{\ensuremath{\mathcal{K}}}

\newcommand{\util}{\ensuremath{\mathbf{U}}}
\newcommand{\pureutil}{\ensuremath{\mathbf{u}}}

\newcommand{\attention}{\ensuremath{\mathbf{h}}}
\newcommand{\ranker}{\ensuremath{\mathbf{R}}}

\newcommand{\finfo}{\ensuremath{\mathcal{F}}}
\newcommand{\einfo}{\ensuremath{\mathcal{E}}}

\newcommand{\spotassigner}{\ensuremath{\mathbf{J}}}

\newcommand{\spotset}{\ensuremath{\mathcal{L}}}

\newcommand{\prodset}{\ensuremath{\mathcal{G}}} 
\newcommand{\spotpotcands}{\ensuremath{\mathbf{C}}} 
\newcommand{\agg}{\ensuremath{\text{agg}}} 

\newcommand{\randgen}{\ensuremath{\mathbf{B}}}

\newcommand{\amax}{\ensuremath{\text{argmax}}}

\newcommand{\define}{\ensuremath{\stackrel{\Delta}{=}}}
\newcommand{\disteq}{\ensuremath{\stackrel{\text{dist}}{=}}}
\newcommand{\dist}{\ensuremath{\text{dist}}}

\newcommand{\prob}[2][]{\ensuremath{\mathbb{P}_{#1}\left( #2 \right)}}
\newcommand{\expected}[2][]{\ensuremath{\mathbb{E}_{#1}\left[ #2  \right]}}

\newcommand{\indd}[2][]{\ensuremath{\mathbf{I}_{#1}\left\{ #2 \right\}}}
\newcommand{\ixj}{\ensuremath{I(x,j)}}
\newcommand{\ixjp}{\ensuremath{I(x,j+1)}}
\newcommand{\iyj}{\ensuremath{I(y,j)}}
\newcommand{\iyjp}{\ensuremath{I(y,j+1)}}
\newcommand{\convcore}[2]{\ensuremath{\pi}^{#1}_{#2}} 
\newcommand{\mass}[1]{\ensuremath{\delta\left\{ #1 \right\}}}
\newcommand{\sizeof}[1]{\ensuremath{\lvert#1\rvert}}

\newcommand{\unicorn}{UniCoRn}

\DeclarePairedDelimiterX{\expectarg}[1]{[}{]}{%
  \ifnum\currentgrouptype=16 \else\begingroup\fi
  \activatebar#1
  \ifnum\currentgrouptype=16 \else\endgroup\fi
}

\newcommand{\innermid}{\nonscript\;\delimsize\vert\nonscript\;}
\newcommand{\activatebar}{%
  \begingroup\lccode`\~=`\|
  \lowercase{\endgroup\let~}\innermid 
  \mathcode`|=\string"8000
}

\title{\LARGE Producer-Side Experiments Based on Counterfactual Interleaving Designs for Online Recommender Systems}
\author{Yan Wang}
\author{Shan Ba}

\affil{LinkedIn Corporation
}
\begin{document}

\maketitle

\begin{abstract}

Recommender systems have become an integral part of online platforms, providing personalized recommendations for purchases, content consumption, and interpersonal connections.  These systems consist of two sides: the producer side comprises product sellers, content creators, or service providers, etc., and the consumer side includes buyers, viewers, or customers, etc. 
To optimize online recommender systems, A/B tests serve as the golden standard for comparing different ranking models and evaluating their impact on both the consumers and producers. 
While consumer-side experiments is relatively straightforward to design and commonly employed to assess the impact of ranking changes on the behavior of consumers (buyers, viewers, etc.), designing producer-side experiments for an online recommender/ranking system is notably more intricate because producer items in the treatment and control groups need to be ranked by different models and then merged into a unified ranking to be presented to each consumer. 
Current design solutions in the literature are ad hoc and lacking rigorous guiding principles.    
In this paper, we examine limitations of these existing methods and propose the principle of consistency and principle of monotonicity for designing producer-side experiments of online recommender systems. Building upon these principles, we also present a systematic solution based on counterfactual interleaving designs to accurately measure the impacts of ranking changes on the producers (sellers, creators, etc.).

\end{abstract}


\section{Introduction}

Recommender systems are ubiquitous in online platforms such as Amazon, Facebook, LinkedIn and Airbnb for suggesting items to buy, contents to view, people to connect, rooms to book, etc. A recommender system has two sides: a \emph{producer side} (e.g., sellers in the marketplace, content creators in feeds, service-providers such as hosts in Airbnb, etc.) and a \emph{consumer side} (e.g., buyers, content viewers, customers, etc.). For each consumer, the recommender system uses machine learning models to rank a set of ``producer items'' (e.g., products from sellers, contents from creators, rooms listed by hosts, etc.) and fills them into pre-designated spots in the user interface of an app or webpage (Figure~\ref{fig_recommender}). The aim of the recommender system is to predict the preference of each consumer and allocate more preferable producer items into the spots where the consumer would pay more attention to.  

\begin{figure}[h]
  \centering
  \includegraphics[width=100mm]{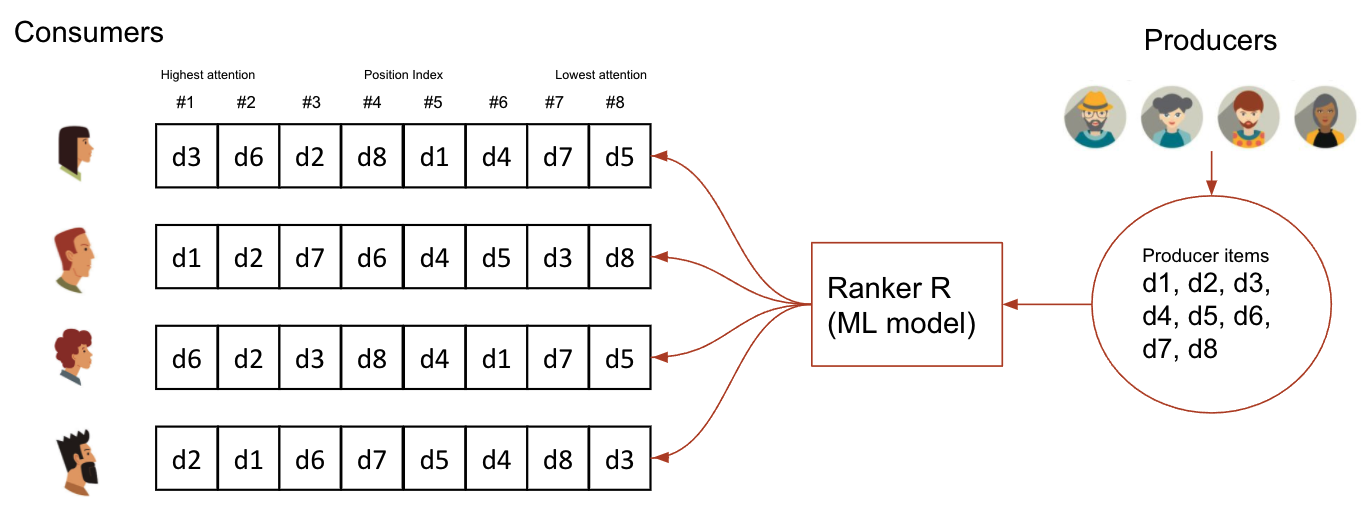}
  \caption{An illustration of the online recommender system.}\label{fig_recommender}
\end{figure}

To optimize an online recommender system, the A/B test (a.k.a. online controlled experiment) ~\citep{xu_infrastructure_2015, tang_overlapping_nodate, bakshy_designing_2014, kohavi_controlled_2009, Kohavi_online_2013, kohavi_tang_xu_2020} is the golden standard for comparing different ranking models and measuring how ranking changes impact the behaviors of consumers/producers. 
In most online platforms, any new ranking model needs to be thoroughly evaluated in online experiments before it can get fully deployed. There are generally two different types of online experiments involved for a recommender system.

\emph{Consumer-side experiments} measure how ranking changes in a recommender impact the behavior of consumers, which are relatively easy to design and widely used in practice. The standard approach is to randomly split all consumers into the control and treatment groups, where each group is associated with a different ranking model (ranker). For each consumer, the recommender ranks \emph{all available producer items} using the variant of ranking model assigned to her group (Figure~\ref{fig_consumer_side}), and then consumer-oriented metrics of the two groups are compared to conclude which ranking model works better for consumers. 

\begin{figure}[h]
  \centering
  \includegraphics[width=100mm]{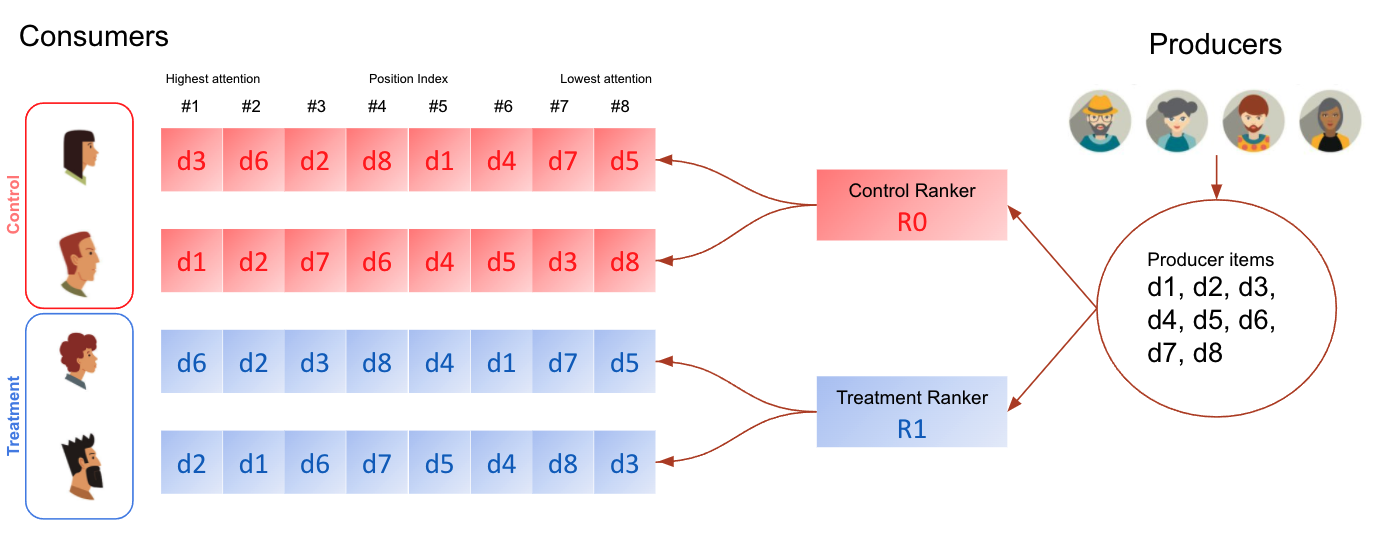}
  \caption{Consumer-side experiment of an online recommender system}\label{fig_consumer_side}
\end{figure}

In addition to the consumer side, it is important to also measure the effects of ranking changes on the producer side, 
because a recommender system needs to be optimized based on objectives derived from both sides.  
For example, purely optimizing the rankings toward buyer's satisfactions in an online marketplace may result in directing most of the traffic to a small portion of top sellers and causing the other sellers (e.g. new sellers) to churn.
Similarly, when improving the ranking model for news feed, we not only need to gauge its impacts on the content viewers, but also need to consider how it would change the behaviors of content creators. 
Nevertheless, despite the importance of measuring producer-side impacts, it is challenging to design \emph{producer-side experiments} for an online recommender system. 
 
\begin{figure}[h]
  \centering
  \includegraphics[width=100mm]{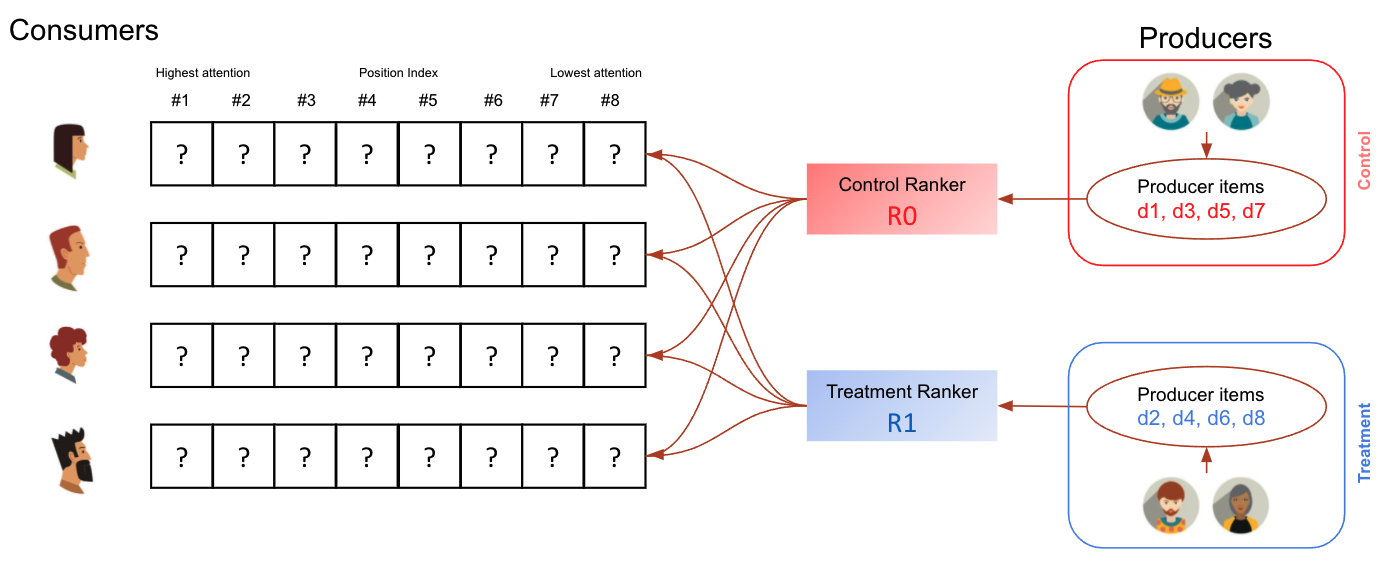}
  \caption{Producer-side experiment of an online recommender system}\label{fig_producer_side}
\end{figure}

The producer-side experiment requires randomly splitting all producers into the control and treatment groups, where producer items from each group are ranked by a different model (Figure~\ref{fig_producer_side}). Because a consumer can only see a single ranked list of producer items each time, designing the producer-side experiment is challenging in that it needs to blend  rankings of producer items from the treatment and control groups together for each consumer. Several design approaches for producer-side experiments have been developed in the literature, but they are ad hoc and suffer from various issues that could lead to biased experiment readouts. 
This paper aims to address this challenge by proposing rigorous design principles that any producer-side experiment should follow in order to accurately measure the effects of ranking changes.
The rest of this paper will be organized as follows. In Section~\ref{sec:wrong_approaches}, we review existing producer-side experiment design solutions and discuss their limitations and biases. Section~\ref{sec:notation} introduces some basic concepts and notations, and in Section~\ref{sec:math_framework}, we propose two general principles for designing producer-side experiments where SUTVA is often violated. 
Building on these principles, Section~\ref{sec:new_principles} derives a rigorous solution based on counterfactual interleaving designs which can ensure an unbiased
comparison between the treatment and control rankers. In Section~\ref{sec:exm}, we provide examples to illustrate the proposed solution, and some final conclusion remarks are given in Section~\ref{sec:conclusions}.

\section{Existing Methods for Designing Producer-Side Experiments}\label{sec:wrong_approaches}

In this section, we provide an overview of the existing methods for designing producer-side experiments and discuss their issues. 
To facilitate the discussion, we will use an illustration example where a total of eight producer items $d_1,\ldots, d_8$ are randomly split into the control group $\prodset_0 = \{d_1, d_3, d_5, d_7\} $ and the treatment group $\prodset_1=\{d_2, d_4, d_6, d_8\}$. The problem of designing producer-side experiments is how to blend the rankings of producer items based on the treatment and control rankers for each consumer as shown in Figure~\ref{fig_producer_side}.

\subsection{Double Randomization}\label{sec:hide}

A straightforward solution is to use double randomization and only show either the treatment or control group of producer items to each consumer. This design requires further splitting consumers into the control and treatment groups where the consumer-side randomization is independent from that at the producer side. If a consumer is in control, the recommender would only show her producer items from the control group $\prodset_0 = \{d_1, d_3, d_5, d_7\} $ that are ranked by the control model. Similarly, if a consumer is in treatment, the recommender would only show her producer items from the treatment group $\prodset_1=\{d_2, d_4, d_6, d_8\}$ that are ranked by the treatment model. See Figure~\ref{fig_hide_items_from_the_other_group} for an illustration. 

Drawbacks of this approach are obvious: the consumer cannot see any producer items from the other group, and each producer item can only be shown to a subset of consumers. These constraints not only lead to poor product experience during the experiment, but they also misrepresent the typical use cases on both the producer side and the consumer side. Consequently, valid conclusions cannot be drawn from such experiments. 

It is important to note that this solution is different from the two-sided randomization design \citep{Johari_2020} and the multiple randomization design \citep{Bajari_2023} in the literature which assume that the intervention can be independently assigned for each consumer-producer pair and thus are not applicable for evaluating the ranking changes in an online recommender system.

\begin{figure}[h]
\centering
\begin{minipage}{.45\textwidth}
  \centering
  \includegraphics[width=.8\linewidth]{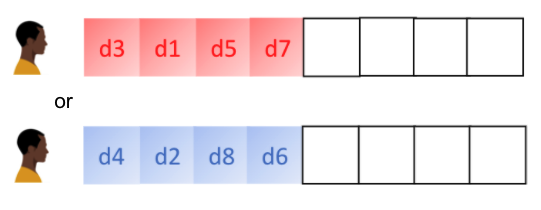}
  \caption{Double randomization and only showing one group of producer items to each consumer. (Top: if the consumer is in control; Bottom: if the consumer is in treatment.)}
  \label{fig_hide_items_from_the_other_group}
\end{minipage} 
\hspace*{0.05cm} 
\begin{minipage}{.45\textwidth}
  \centering
  \includegraphics[width=.8\linewidth]{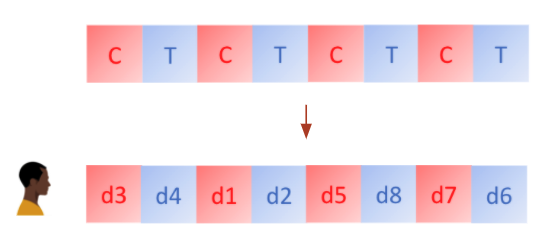}
  \caption{Random spot labeling. (For one consumer, producer items in control are ranked as $[d_3, d_1, d_5, d_7]$ and placed among the control spots, while producer items in treatment are ranked as $[d_4, d_2, d_8, d_6]$ and placed among the treatment spots)}
  \label{fig_random_spot_allocation}
\end{minipage}
\end{figure}

\subsection{Random Spot Labeling} \label{sec:random_spot}
Another common design solution is illustrated in Figure~\ref{fig_random_spot_allocation}. For a consumer, the method first randomly labels each spot in her final ranking list to be either in treatment (T) or control (C). Then, producer items in control are ranked by the control model (e.g., $[d_3, d_1, d_5, d_7]$) and placed among the control spots, while producer items in treatment are ranked by the treatment model (e.g., $[d_4, d_2, d_8, d_6]$) and placed among the treatment spots. 

This design approach is essentially based on a random merger of the treatment and control rankings. It is better than the previous approach in Section~\ref{sec:hide} as all the producer items can be shown to each consumer. However, the merged ranking is still not representative of the real product experience. To see this, consider an AA test scenario where the treatment and control models are the same. In this case, we would expect the final ranking to remain the same as using either treatment or control model to rank all producer items. However, the design method in Figure~\ref{fig_random_spot_allocation} would generate very different ranking results because the design imposes an extra constraint on the treatment or control label of each spot in the final ranking. For example, suppose in the AA test, both treatment and control models would rank the eight producer items as $[d_3, d_1, d_5, d_7, d_4, d_2, d_8, d_6]$ for a consumer. Then, the correct final ranking for this consumer should just be $[d_3, d_1, d_5, d_7, d_4, d_2, d_8, d_6]$ and the corresponding treatment/control label of each spot should be $[C, C, C, C, T, T, T, T]$. The random spot labeling constraint $[C, T, C, T, C, T, C, T]$ in Figure~\ref{fig_random_spot_allocation}, on the other hand, results in an inaccurate final ranking $[d_3, d_4, d_1, d_2, d_5, d_8, d_7, d_6]$ for the consumer, which cannot reflect the real product experience for producers.

\subsection{SUTVA and Counterfactual Rankings}\label{sec:sutva}

The ``Stable Unit Treatment Value Assumption'' (SUTVA) \citep{imbens_rubin_2015} is a standard assumption in designing A/B tests which requires that the potential outcome for one unit in the experiment depends only on its own treatment status and should not be affected by the treatment assignment to the other units. For producer-side experiments, SUTVA means that producers in each treatment group should not be affected by the existence of other treatment group; instead, their behavior should be the same as if the ranking model associated with their group is applied to all of the producers. 

\cite{facebook_2020} defines the \emph{control counterfactual ranking} as the ranking of all producer items (from both the treatment and control groups) based on the control model. It represents the ranking result as if the control ranker is ramped to 100\% of the producers. 
Similarly, the \emph{treatment counterfactual ranking} is defined as using the treatment model to rank all producer items (not only the producer items in the treatment group), which represents the ranking result as if the treatment model is applied to 100\% of the site traffic. Figure~\ref{fig_identical_rankers} and Figure~\ref{fig_non_conflict_merging} give two examples of the counterfactual rankings. 
When merging the rankings of producer items from the treatment and control groups together for each consumer in Figure~\ref{fig_producer_side}, SUTVA requires that producer items from the control group should be placed in the same positions as if they were in the control counterfactual ranking while producer items from the treatment group should be placed in the same positions as if they were in the treatment counterfactual ranking. We call such a merged ranker in the producer-side experiment as the \emph{SUTVA ranker} $\ranker_*$. It has the desirable property that all producer items of each group are placed at the same positions as if the corresponding ranking model is ramped to 100\% of the site traffic.

In practice, a valid SUTVA ranker may not exist. For designing producer-side experiments, we summarize the following two basic rules for when SUTVA can be met and must be followed. First is for the AA-Test scenario that we have described at the end of Section~\ref{sec:random_spot}.   
\begin{abc}[AA-Test Scenario]
  \label{principle:identical-ranker}
   If the treatment and control rankers are identical, their merged ranker in the producer-side experiment should remain the same as the original ranker which is also the SUTVA ranker. (Figure~\ref{fig_identical_rankers})
\end{abc}
This rule is important because in practice the difference between treatment and control rankers are often small (i.e. small treatment effect). All valid design methods for producer-side experiments need to satisfy Rule~\ref{principle:identical-ranker} and correctly yield the SUTVA ranker in the AA-test scenario.

\begin{figure}[h]
\centering
\begin{minipage}{.5\textwidth}
  \centering
  \includegraphics[width=.9\linewidth]{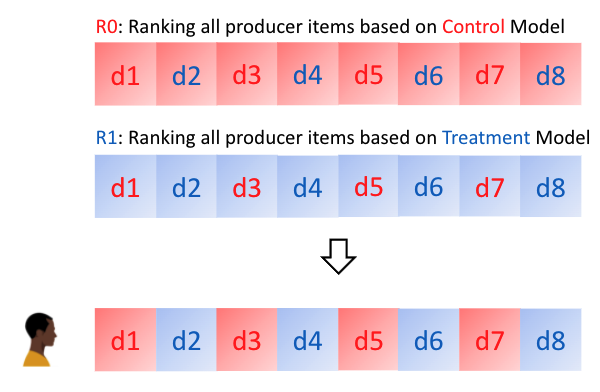}
  \caption{AA-Test Scenario}
  \label{fig_identical_rankers}
\end{minipage}%
\begin{minipage}{.5\textwidth}
  \centering
  \includegraphics[width=.9\linewidth]{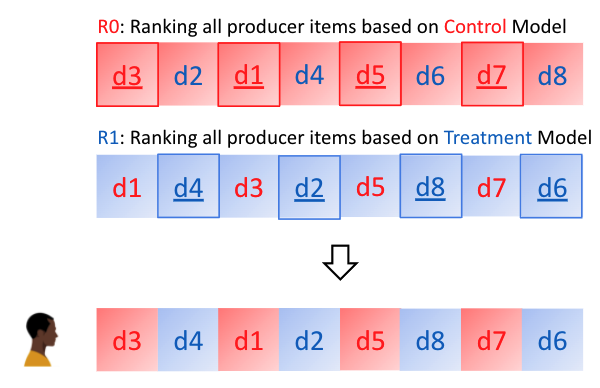}
  \caption{Non-Conflict-Merging Scenario}
  \label{fig_non_conflict_merging}
\end{minipage}
\end{figure}

Our next rule summarizes the non-conflict merging scenarios in which a valid SUTVA ranker $\ranker_*$ exists and should always be used. 
\begin{abc}[Non-Conflict-Merging Scenario]
  \label{principle:non_conflict}
   When the treatment and control counterfactual rankings have no merging conflicts, the merged ranker in the producer-side experiment should be uniquely determined by SUTVA and place every producer item into the position as if the ranking model associated with its group is applied to all producers. (Figure~\ref{fig_non_conflict_merging})
\end{abc}

When there are conflicts in merging the treatment and control counterfactual rankings (i.e., producer items from different groups both demand the same position in the merging process), SUTVA cannot be perfectly met and this is the challenging part in designing producer-side experiments. 

To facilitate the discussions in the rest of this paper, we define a ranker $\ranker$ as a ranking model which provides a ranked list of producer items to a consumer in each session  of the recommender system. 
Mathematically, $\ranker$ generates an one-to-one mapping function in each session:
\begin{equation}
 \ranker: \prodset=\left\{ d_1, d_2, \dots \right\} \longrightarrow   \spotset\define\left\{ 1, 2, \dots, \sizeof{\prodset} \right\}
\end{equation}
where $\prodset = \left\{ d_1, d_2, \dots \right\}$ denotes the set of all producer items, 
$\spotset$ is a set of ranks each of which corresponds to a spot in the consumer's user interface. Here, a smaller rank value represents a better match between the producer item and the consumer, and without loss of generality, we assume that the consumer's user interface has $\sizeof{\prodset}$ spots where spots with smaller indices tend to receive more attentions from the consumer (i.e., spots at the top of the page). The recommender system fills in these spots by matching the ranks of the producer items with the spot indices, i.e., the producer item $d$ with $\ranker(d) = i$ will be put in spot $i$.  

Let $\ranker_0$ denote the control counterfactual ranker, $\ranker_1$ denote the treatment counterfactual ranker, and then the SUTVA ranker $\ranker_*$ can be represented as
\begin{equation}
\label{equ:optimal_self_contradicting_design_R}
\ranker_*(d) = \left\{ \begin{array}{ll}
\ranker_0(d) & \textrm{if $d$ is in the control group}\\
\ranker_1(d) & \textrm{if $d$ is in the treatment group}
\end{array} \right.
\end{equation}
for all producer items $d \in \prodset$.

The SUTVA ranker $\ranker_*$ is the optimal solution for designing producer side experiment as long as it exists (no merging conflicts). However, if there exists $d, d'\in\prodset$, $d\neq d'$ such that $\ranker_*(d) = \ranker_*(d')$, the SUTVA ranker $\ranker_*$ is not a valid ranker as producer items $d$ and $d'$ are demanding the same position in the merged ranking. In the next two sections, we review existing solutions in the literature to handle such merging conflicts, discuss their shortcomings and also motivate our proposed principles.

\subsection{Counterfactual Interleaving Design}\label{subsec:facebook}

In this paper, we will use the \textit{counterfactual interleaving design} to refer to the design of producer-side experiments based on merging (or interleaving) different counterfactual rankings.
It is important to distinguish it from the traditional interleaving designs \citep{radlinski_interleaving_2013, netflix_interleaving_2017, airbnb_interleaving_2022} for consumer-side experiments which are not based on the counterfactual rankings.

\cite{facebook_2020} from the Facebook Marketplace proposed counterfactual interleaving designs which only randomly label a small percent (e.g., 1\%) of producers as the control and treatment groups to minimize the chances of having merging conflicts and avoid the challenges in resolving merging conflicts in $\ranker_*$. The rest of the producers would still be shown in the recommender, but their metrics would not be included in the experiment analysis. This approach can be summarized as follows:

Step 1: Generate counterfactual rankings $\ranker_0$ and $\ranker_1$ as if the control or treatment model is ramped to 100\% of producers. 
 
Step 2: Merge $\ranker_0$ and $\ranker_1$ into the SUTVA ranker $\ranker_*$:
\begin{itemize}
\item For producer items in the control group, get their positions from the control counterfactual ranking $\ranker_0$.
\item For producer items in the treatment group, get their positions from the treatment counterfactual ranking $\ranker_1$.
\item For the rest of producer items that are neither in the treatment nor control groups, get their positions from either $\ranker_0$, $\ranker_1$ or another ranker, but these producer items would not be included in the treatment v.s. control comparison of the experiment.
\end{itemize}

Step 3: In case that a pair of producer items demand the same position in $\ranker_*$ (i.e., there exists $d, d'\in\prodset$, $d\neq d'$ such that $\ranker_*(d) = \ranker_*(d')$), simply decide their order randomly.

Figure~\ref{fig_facebook} provides a simple illustration of this approach. Because both treatment and control groups only contain a small fraction (e.g., 1\%) of producer items, \cite{facebook_2020} shows that the probability of having merging conflicts in Step 2 is very low. In case that a pair of producer items happen to demand the same position in $\ranker_*$, the method can randomly decide their order with equal probabilities. 
Because such merging conflicts are rare, the final ranker has the advantage that it is approximately a SUTVA ranker where the ranking position of each producer item does not depend on what ranker is applied to the other producer items.  

Obviously, downside of this counterfactual interleaving design is that its lower ramp \% of producer items vastly limits the experiment power. Although power may not be a concern for Facebook which has enormous amounts of online traffic, the method is not applicable to many other online recommender systems due to lack of power.

\begin{figure}[h]
\centering
\begin{minipage}{.5\textwidth}
  \centering
  \includegraphics[width=.9\linewidth]{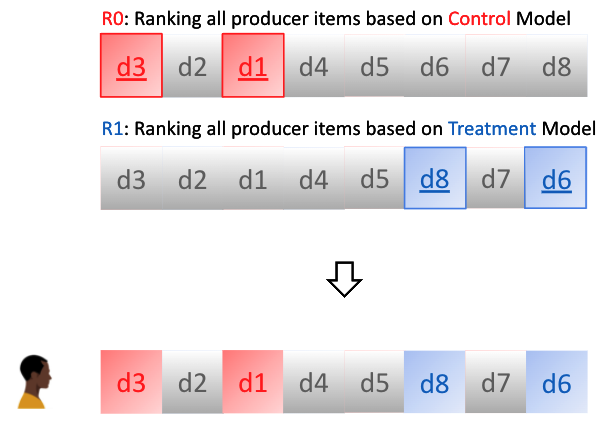}
  \caption{Counterfactual interleaving design from \cite{facebook_2020}}
  \label{fig_facebook}
\end{minipage}%
\begin{minipage}{.5\textwidth}
  \centering
  \includegraphics[width=.9\linewidth]{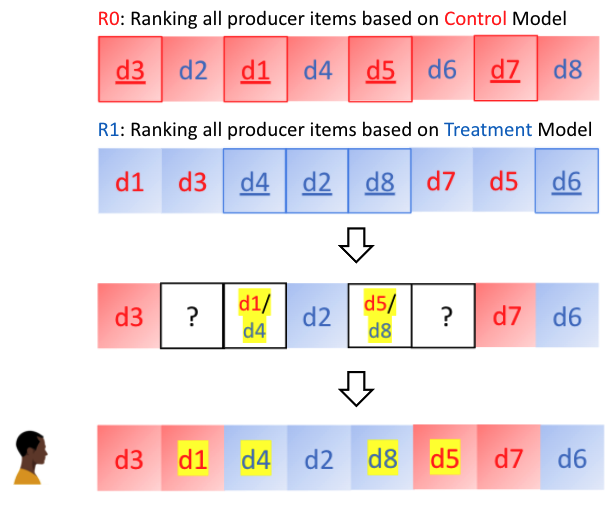}
  \caption{Counterfactual interleaving design with merging conflicts}
  \label{fig_unicorn}
\end{minipage}
\end{figure}

\cite{nandy2021unicorn} proposed the Unifying Counterfactual Rankings (\unicorn) approach, which is the same as  \cite{facebook_2020}'s counterfactual interleaving design except for allowing larger \% of producers to be included in the treatment and control groups to increase the power of the experiment. 
When the treatment and control counterfactual rankings have merging conflicts (i.e., there exists $d, d'\in\prodset$, $d\neq d'$ such that $\ranker_*(d) = \ranker_*(d')$) and the SUTVA ranker $\ranker_*$ is not a valid ranker, the UniCoRn ranker also chooses to resolve any merging conflicts randomly to ensure that the final ranker is a valid ranker. 
\cite{nandy2021unicorn} showed that the UniCoRn ranker is the ranker that gets closest to $\ranker_*$ in terms of the sum of squared error distance. 

Nevertheless, both \cite{facebook_2020} and \cite{nandy2021unicorn}'s solutions are ad hoc and have not addressed the critical question of how to correctly resolve merging conflicts in the counterfactual interleaving designs.  For example, when two producer items $d$ and $d'$ have merging conflict, how to determine the probability that $d$ should be placed ahead of $d'$? 
Only minimizing the sum of squared error distance $\dist(\ranker, \ranker_*)\define\sum_{d\in\prodset} \left( \ranker(d) - \ranker_*(d) \right)^2$ cannot determine the probabilities because any tie-breaking probabilities would lead to the same $\dist(\ranker, \ranker_*)$. In fact, choosing different tie-breaking probabilities could generate a large number of possible UniCoRn rankers with different distributions of the producer items, but, unfortunately, the majority of them would lead to biased comparisons between the treatment v.s. control groups. 

Consider the example of simple random tie-breaking (with equal probabilities) from \cite{facebook_2020}  and \cite{nandy2021unicorn}. 
If the treatment group is ramped at a small percent of traffic (e.g. 5\% of producers) while the control group is ramped at a large percent of traffic (e.g. 95\% of producers), the beginning part of the treatment counterfactual ranking $\ranker_1$ would contain very few treatment producer items. When equal probabilities are used to randomly break the ties in merging $\ranker_0$ and $\ranker_1$, it would be hard to have any treatment producer items to appear in the right positions (as determined by $\ranker_1$) in the beginning of the merged ranking where viewers mainly pay attention to. 
Consequently, the treatment counterfactual ranking $\ranker_1$ cannot be properly represented in the merged ranker and the experiment readouts would be biased against the treatment group. 
In order to accurately measure the effects of ranking changes on the producer side, we will develop two general principles for designing producer-side experiments in the following sections. 
Based on the proposed principles, we will show that when resolving the merging conflicts between $\ranker_0$ and $\ranker_1$, an unbiased counterfactual interleaving design need to follow a rigorous procedure to preserve the relative order of the SUTVA ranker and assign items from the smaller treatment group with a higher probability to be placed in the correct position (determined by its own counterfactual ranking). We will also show that the optimal order of conflicting producer items in some cases should be determined deterministically instead of randomly.

\section{Notations and Concepts} \label{sec:notation}

In this section, we define some key concepts and mathematical notations in producer-side experiments, which lays the groundwork for subsequent discussions on the design principles. 
Let $k \in \mck$ represent different treatment variants in the producer-side experiments. For example, $\mck = \{0,1\}$ where $k=0$ represents the control variant and $k=1$ represents the treatment variant. Suppose each producer item in $\prodset$ is randomly assigned into one of the treatment groups $\prodset_k$ with probability $p_k$.

\subsection{Experiment Readouts and Counterfactual Readouts}\label{subsec:cand_level_metrics_agg}

We first define the outcomes of producer-side experiments based on metrics aggregated at the producer (item) level.
Let $\finfo$ represent the set of information pertaining to the consumer and producer items' features, and let $\util(d)$ denote the outcome of a metric associated with producer item $d$. Typically, $\util(d)$ depends on $\finfo$ as well as the ranking position of the producer item, i.e., $\util(d) = \util(d; \ranker(d), \finfo)$. Without loss of generality, we assume the metric $\util$ is larger the better.

Let $\agg_k^*(\util)$ represent the \textit{counterfactual readout} which is an aggregation of metrics of all producer items under the assumption that ranker $\ranker_k$ is applied to 100\% of producers. This can be formally expressed as:
\begin{equation}\label{equ:agg_k_star_U}
  \agg_k^*(\util)\define\sum_{d\in \prodset} \util\left(d; \ranker_k\right).
\end{equation}

Ideally, we want to compare the counterfactual readouts $\left\{ \agg_k^*(\util) \right\}$ for various $k$ to determine the optimal ranker, but $\left\{ \agg_k^*(\util) \right\}$ are not observable from the experiment. Instead, we can only observe the \textit{experiment readout} which is an aggregation of $\util(d)$ for a treatment group $k$ in the experiment:
\begin{equation}
  \label{equ:aggregate_metrics}
  \agg_k\left( \util \right) \define \frac{1}{p_k} \sum_{d\in \prodset_k} \util\left(d; \ranker\right).
\end{equation}
Here $\ranker$ is the final merged ranker in the experiment and the coefficient $1/p_k$ accounts for the fact that each producer item has a probability $p_k$ of being included in $\prodset_k$.

Both $\agg_k$ and $\agg^*_k$ are random variables and we are particularly interested in their expected values. It is important to acknowledge that their randomness originates from multiple sources. One source is the feature set $\finfo$, and another source relates to the experimental design, i.e., how items are randomly allocated to treatments and how the merged ranker resolves merging conflicts. We will denote this experimental information by $\einfo$ and postulate the following assumption:
\begin{assumption}\label{assmp:e_f_independent}
  The distributions of $\finfo$ and $\einfo$ are independent.
\end{assumption}

We generally lack knowledge about the specifics of the distribution of $\finfo$, which is influenced by the complex interactions between consumers and producer items. However, we have complete knowledge regarding the distribution of $\einfo$, as it is determined by the experimental design. Consequently, expectations will always be taken with respect to $\einfo$ and conditioned on $\finfo$. This is denoted by the operator $\expected[\einfo]{\cdot}$ or $\expected{\cdot\vert\finfo}$.

Moreover, it is crucial to note that the experimental-related randomness from $\einfo$ is only present in $\agg_k(\util)$, while $\agg^*_k(\util)$ is measurable with respect to $\finfo$, i.e.,
\begin{align*}
  \expected{\agg^*_k(\util)\vert \finfo} = \agg^*_k(\util),
\end{align*}
and the expectation with respect to $\einfo$ only needs to be considered for $\agg_k(\util)$.

Let us define
\begin{align}
  \label{equ:agg_k_bar}
  \overline{\agg_k}(\util) \define \expected{\agg_k(\util)\vert\finfo},
\end{align}
as the expected value of the experiment readout, conditioned on $\finfo$. Ideally, we want $\overline{\agg_k}(\util)$ to be an unbiased estimator of the counterfactual readout $\agg_k^*(\util)$, which requires SUTVA to be met. Unfortunately, as explained in the previous section, SUTVA is often violated in the producer-side experiments due to merging conflicts and this is why we need to develop new design principles to ensure that valid conclusions can be drawn from producer-side experiment readouts. 

\subsection{Attention Functions and Convoluted Attention Functions} \label{subsec:attnFun}
In the user interface of a recommender system, different positions or spots receive varying degrees of visibility or attention from the consumers. For any position $j$, let $\attention(j)\geq 0$ represent the amount of attention garnered by a producer item at that position, where the \textit{attention function} $\attention(j)$ is monotonic decreasing as positions in the recommender system are indexed in such a way that smaller indices receive more attention from the consumer. For example, top ranking spot ($j=1$) in the recommender receives the highest attention from the consumer and the spots at the end receive little attention as few consumers would scroll far down the page. 

Based on the attention function, we introduce an assumption regarding the structure of the observed metric outcome $\util$ associated with each producer item.

\begin{assumption}
\label{assmp:utility_function_decompose_utility_and_attention}
For producer item $d$, its metric outcome $\util(d)$ can be decomposed into a product of a pure metric $\pureutil(d)$ representing the inherent utility of $d$, which is independent of the ranking of the producer items, and an attention function $\attention$ that solely depends on the rank or position $\ranker(d)$ of the producer item. Formally, we have:

  \begin{equation*}
    \pureutil: \prodset\to[0,\infty),  \attention: \spotset\to[0, \infty),
  \end{equation*}
  and
  \begin{equation}
    \label{equ:utility_function_decompose_utility_and_attention}
    \util(d) = \pureutil(d)\times\attention\circ\ranker(d).
  \end{equation}
\end{assumption}

In the equation above, we use $f \circ g$ to denote the composition of functions $f$ and $g$, such that $f \circ g(d) = f[g(d)]$.
This decomposition reflects the fact that a producer item's metric can be viewed as the product of its inherent quality (captured by $\pureutil$) and the attention it receives based on its position (captured by $\attention$).

Let us consider for a given position $j$ and item $x$ such that $\ranker_0(x) = j$. When $x$ is in the control group $\prodset_0$, the merged ranker $\ranker$ may not always place $x$ at position $j$ due to merging conflicts. Instead, the final position of $x$ under $\ranker$ can vary, and follows a certain distribution. Let us denote this distribution by $\convcore{0}{j}$, such that for any position $j'$,
\begin{equation*}
	\convcore{0}{j}(j') = \prob{\ranker(x) = j'\vert x\in\prodset_0}.
\end{equation*}
Due to such randomness, the average level of attention that the producer item $x$ receives when $x\in\prodset_0$, is not strictly $\attention(j)$, but rather a weighted average of $\attention(j')$, taking into account the probabilities $\left\{ \convcore{0}{j}(j') \right\}$. Let us define this “average attention” as $\attention^0(j)$:
\begin{equation} \label{equ:average_attention}
	\attention^0(j)\define \sum_{j'}\convcore{0}{j}(j')\attention(j').
\end{equation}
Since $j$ can be any position, equation \eqref{equ:average_attention} essentially gives rise to a new “attention function” $\attention^0$, which we will refer to as the \textit{convoluted attention function}. The effect of resolving merging conflicts in $\ranker$ can be conceptualized as a transformation of the underlying attention function from $\attention$ to $\attention^0$ for the control group and its ranker $\ranker_0$.
Similarly, this can also be applied to the treatment group and its ranker $\ranker_1$: For any position $j$, define
\begin{equation*}
	\convcore{1}{j}(j') = \prob{\ranker(y) = j'\vert y\in\prodset_1},
\end{equation*}
and
\begin{equation*} 
	\attention^1(j)\define \sum_{j'}\convcore{1}{j}(j')\attention(j').
\end{equation*}
The resolution of merging conflicts by $\ranker$ essentially imposes a transformation of the attention function from $\attention$ to $\attention^1$ for the ranker $\ranker_1$.

In general, consider a treatment $k \in \mck$, a spot $j \in \spotset$, and the merged ranker $\ranker$ that is utilized as the final ranker in the producer-side experiment as described in equation \eqref{equ:aggregate_metrics}. Let $\spotassigner$ be the inverse of ranker function $\ranker$:
\begin{equation*}
	\spotassigner(j) = d \text{ if and only if } \ranker(d)  = j.
\end{equation*}
which is a spot-filling function that maps a spot index $j$ back to the producer item that occupies it. 
Suppose the producer item $d$ would occupy spot $j$ according to ranker $\ranker_k$: $\ranker_k(d) = j$ or $d = \spotassigner_k(j)$. Due to merging conflicts as discussed in Section~\ref{sec:sutva}, the final merged ranker cannot guarantee $\ranker(d) = j$. Instead, $\ranker(d)\vert_{d \in \prodset_k}$ is random and let $\convcore{k}{j}$ represent the distribution of this random variable:
\begin{align}
  \label{equ:def_conv_core}
  \convcore{k}{j}(j') = \prob{\ranker\circ\spotassigner_k(j) = j'\vert\spotassigner_k(j)\in\prodset_k}.
\end{align}
Obviously, the forms of convolution kernels $\left\{ \convcore{k}{j} \right\}$ are determined by the rankers $\left\{\ranker_k\right\}$ and way they are merged into $\ranker$. They do not depend on the specific form of attention functions.

Define $\attention^k$ as the convolution of the probability family $\left\{ \convcore{k}{j} \right\}$ with the attention function $\attention$ and call $\left\{ \convcore{k}{j} \right\}$ the set of convolution kernels. For any spot $j \in \spotset$, we can formally express the convoluted attention function as:
\begin{align}
  \label{equ:def_h_conv}
  \attention^k(j) = \sum_{j'\in\spotset} \convcore{k}{j}(j')\times\attention(j').
\end{align}
The convolution kernels $\left\{ \convcore{k}{j} \right\}$ quantifies the deviation of the merged ranker $\ranker$ from the SUTVA ranker $\ranker_*$. In the specific case where $\ranker = \ranker_*$ (e.g., an A/A test scenario), $\convcore{k}{j}$ becomes $\mass{j}$ (i.e., the distribution concentrated on the single spot $j$), and consequently, $\attention^k$ simplifies to $\attention$.

\section{New Design Principles}\label{sec:math_framework}

As we have discussed in Section~\ref{sec:wrong_approaches}, currently there are no rigorous guiding principles available for designing producer-side experiments when SUTVA is violated, and the existing solutions in the literature are ad hoc which would lead to biased designs. In this section, we will develop general principles for designing producer-side experiments that are essential for any design solutions to follow.

\subsection{Principle of Consistency}  \label{subsec:rule_merger}

For a randomized experiment, a fundamental requirement is that the treatment and control groups need to be comparable and the only expected difference between them is caused by the intervention (i.e., different rankers) being studied. 
In producer-side experiments, however, the existing ad hoc solutions often introduce other confounding factors between the treatment and control groups (i.e., receiving different amounts of attentions from consumers) which would bias the experiment results.  
For example, consider the top-ranking position (i.e., spot 1) in the recommender which receives the highest attention from consumers. Assume $\ranker_0$ ranks item $x$ as the top candidate while $\ranker_1$ ranks item $y$ as the top candidate: $\ranker_0(x) = \ranker_1(y) = 1$. Due to the merging conflict, when $x \in \prodset_0$, the spot for $x$ in the final merged ranker $\ranker(x)\vert_{x\in\prodset_0}$ will be random, where the randomness is determined by the specific method with which $\ranker$ merges $\ranker_0$ and $\ranker_1$. 
Suppose for $x\in\prodset_0$, it has $80\%$ chance of being placed at the best spot (spot $1$ in $\ranker$) and for $y\in\prodset_1$, it has a 30\% chance of being ranked first in $\ranker$.
Clearly, this design is biased in favor of $\ranker_0$: the best candidate $x$ according to $\ranker_0$ has a higher chance of getting the top spot compared to the best candidate $y$ according to $\ranker_1$. Such bias will be reflected in the final experiment outcome and be confounded with the treatment effect under study.

To ensure an apple-to-apple comparison in the producer-side experiments, a fair design should require that $\ranker(x)|_{x \in \prodset_0}$ and $\ranker(y)|_{y \in \prodset_1}$ have the same distribution. This generalizes beyond just the top spot and applies for any integer $j > 0$: when $x\in\prodset_0$, it should have the same chances to be placed at spot $j$ ($j = 1, 2, \ldots$) as $y$ does under $y\in\prodset_1$.  In other words, in order for the treatment and control groups to be comparable under the merged ranker $\ranker$, we need to require that, for any ranking spot $j$, and for any producer items $x$ and $y$ such that $\ranker_0(x) = \ranker_1(y) = j$, the distributions of $\ranker(x)|_{x \in \prodset_0}$ and $\ranker(y)|_{y \in \prodset_1}$ should be identical.
This requirement will be referred to as the \textit{Principle of Consistency}, which is formally defined below. 
\begin{principle}[Principle of Consistency]
  \label{def:conv_attn_consistent}
In designing producer-side experiments, the convolution kernels $\left\{ \convcore{k}{j} \right\}$ defined in equation~\eqref{equ:def_conv_core}, which represent the distributions of $\ranker(d)\vert_{d \in \prodset_k}$ for any given spot $j$,  should be invariant with respect to $k$.
\end{principle}

This principle can be formally justified based on the mathematical framework defined in Section~\ref{sec:notation}. 
Based on equations \eqref{equ:agg_k_star_U} and \eqref{equ:utility_function_decompose_utility_and_attention}, the counterfactual readout $\agg^*_k(\util)$  can be expressed as:
\begin{align}
  \label{equ:agg_star_k_util_decomp}
  \begin{split}
    \agg^*_k(\util) & = \sum_{d\in \prodset} \pureutil(d)\times\attention\circ\ranker_k(d) \\
  & = \sum_{j\in \spotset} \pureutil\circ\spotassigner_k(j)\times\attention(j)
  \end{split}
\end{align}
where $\spotassigner_k(j) = d \text{ if and only if } \ranker_k(d)  = j$.
It can be seen that $\agg^*_k(\util)$ depends on $k$ only through the different rankers $\ranker_k$ (or $\spotassigner_k$) while the attention function $\attention$ is the same for different $k$.
Our next theorem below shows that this ideal property cannot be guaranteed in the observed experiment readouts $\left\{ \agg_k(\util) \right\}$ or their expected values $\left\{ \overline{\agg_k}(\util) \right\}$. 
\begin{theorem} \label{thm:conv_attn}
Under Assumptions \ref{assmp:e_f_independent} and \ref{assmp:utility_function_decompose_utility_and_attention}, for any treatment $k \in \mck$, $\overline{\agg_k}(\util)$ can be computed as follows:
  \begin{align}
    \label{equ:change_of_attn}
\overline{\agg_k}(\util)  =  \sum_{j\in \spotset} \pureutil\circ\spotassigner_k(j)\times\attention^k(j) 
  \end{align}
\end{theorem}

Proof of this theorem is given in Appendix \ref{app:proof_thm_conv_attn}.
By comparing Equations \eqref{equ:agg_star_k_util_decomp} and \eqref{equ:change_of_attn}, we can see that the difference between $\agg^*_k(\util)$ and $\overline{\agg_k}(\util)$ is effectively a modification of the attention function through convolution, denoted as $\attention\to\attention^k$. Furthermore, \eqref{equ:change_of_attn} shows that for various $k$, differences in $\overline{\agg_k}(\util)$ not only are due to the differences in rankers $\ranker_k$ (or $\spotassigner_k$) but they can also be caused by different convoluted attention functions $\attention^k$. This makes it indiscernible whether the disparities in the expected experiment readouts $\left\{ \overline{\agg_k}(\util) \right\}$ between the treatment and control groups stem from the rankers or their attention functions.
To mitigate this confounding ambiguity, a correctly designed producer-side experiment must ensure that the convoluted attention functions $\attention^k$ remain independent of $k$ (i.e., while treatment and control groups correspond to different rankers, they must share the same convoluted attention function to be comparable). Given that the specific form of the attention function $\attention$ is unknown, the only way to assure independence of $\attention^k$ on $k$ is by requiring that the convolution kernel $\convcore{k}{j}$ does not rely on $k$ in equation \eqref{equ:def_h_conv}, which formalizes the principle of consistency above.

\subsection{Principle of Monotonicity}

In addition to the consistency principle, in this section we introduce another important principle for designing producer-side experiments. 
Because the spots or positions in a recommender system are indexed according to the level of attention they receive (i.e., position 1 receives the highest attention, position 2 the second highest, etc.), the attention function $\attention$ is inherently defined to be monotonically non-increasing: $\attention(1)\geq \attention(2) \geq \dots$ and the  recommender is designed to place the most suitable (highest ranked) producer item in position 1, the second best in position 2, and so on. However, due to the randomness from resolving merging conflicts in the producer-side experiments, the average level of attention that a producer item receives is not strictly $\attention$ but a weighted average of $\attention$, which is defined as the convoluted attention $\attention^k$ in Section~\ref{subsec:attnFun}. As a result, for producer-side experiments to be valid, we not only need to have monotonically non-increasing attention function $\attention$, but also need to require the convoluted attention function $\attention^k$ to retain this monotonic characteristic. 
We will refer to this requirement as the Principle of Monotonicity.  

Since the exact form of the attention function $\attention$ is unknown, the design of producer-side experiments can leverage the convolution kernels $\convcore{k}{j}$ in equation~\eqref{equ:def_h_conv} to ensure $\attention^k$ is decreasing. 
Let $F_{\pi}(x)$ denote the cumulative distribution function (CDF) of a distribution $\pi$ on real numbers, i.e.,
\begin{equation*}
  F_{\pi}(x) \define \prob[\pi]{X\leq x}
\end{equation*}
where $X$ follows the distribution $\pi$. We can then define the partial order relation as
\begin{equation*}
  \pi_1 \prec \pi_2 \Leftrightarrow \forall x, F_{\pi_1}(x) \geq F_{\pi_2}(x).
\end{equation*}
It is a well-established fact that for any monotonically decreasing function $\attention$ on the real line, if $\pi_1 \prec \pi_2$, then $\expected[\pi_1]{\attention} \geq \expected[\pi_2]{\attention}$.
We can now formally introduce the Principle of Monotonicity as follows:
\begin{principle}[Principle of Monotonicity]
  \label{def:monotocity}
  In designing producer-side experiments, the convoluted attention functions $\attention^k$ must be monotonically non-increasing for any $k$ and for any non-increasing attention function $\attention$. Equivalently, the convolution kernels $\left\{ \convcore{k}{j} \right\}$ defined in equation~\eqref{equ:def_conv_core} needs to be non-decreasing with respect to $j$, for any $k$. 
\end{principle}

Using the mathematical framework defined in Section~\ref{sec:notation}, we can provide some further justifications of this monotonicity principle.

\begin{lemma}
  \label{lem:util_attn_line_up}
  Define $\ranker_{\pureutil}$ as a ranker that ranks producer items $d \in \prodset$ according to their pure metric $\pureutil(d)$ (in descending order). Then, ranker $\ranker_{\pureutil}$ satisfies:
  \begin{equation*}
    \ranker_{\pureutil} 
    = \amax_{\ranker}\sum_{d\in\prodset}\pureutil(d)\times\attention\circ\ranker(d).
  \end{equation*}
\end{lemma}

Proof of this Lemma is given in Appendix~\ref{app:proof_util_attn_line_up}. 
Lemma \ref{lem:util_attn_line_up} demonstrates that the ranker which maximizes the counterfactual readout $\agg^*(\util)= \sum_{d\in\prodset}\util(d) = \sum_{d\in\prodset}\pureutil(d)\times\attention\circ\ranker(d)$ aligns precisely with the ranking based on producer item's pure metric $\pureutil(d)$. This provides a fundamental justification for why recommendation systems aim to model producer item's intrinsic utility $\pureutil(d)$ and use the scores obtained from these models to establish ranking. However, it is crucial to recognize that this alignment hinges on the attention function being monotonically non-increasing.
Building upon Lemma \ref{lem:util_attn_line_up}, we can have the following corollary.

\begin{corollary} \label{cor:maximizer_counterf} 
Assuming there are $K$ rankers $\ranker_0$, ..., $\ranker_{K-1}$ being compared in a producer-side experiment, and one of them, say $\ranker_k$, is equivalent to $\ranker_{\pureutil}$ as defined in Lemma \ref{lem:util_attn_line_up}. If the attention function $\attention$ is monotonically non-increasing, then the counterfactual readout $\agg^*_k(\util)$ for group $k$ is superior to those of other groups, meaning
	\begin{equation*}
		\forall k'\neq k, \agg^*_k(\util) \geq \agg^*_{k'}(\util).
	\end{equation*}
\end{corollary}

Corollary \ref{cor:maximizer_counterf} indicates that the counterfactual readout $\agg^*(\util)$ from producer-side experiments can be used to correctly identify the best ranker if the attention function is monotonically non-increasing. 
Nevertheless, in practice we cannot directly observe the counterfactual readout $\agg^*(\util)$. 
The following key corollary, which is derived based on the observed experiment readouts $\left\{ \agg_k(\util) \right\}$ or their expected values $\left\{ \overline{\agg_k}(\util) \right\}$, highlights the importance of having both the consistency and monotonicity principles in designing producer-side experiments:
\begin{corollary}
  \label{cor:consst_mono_ub_mm}
Assume $K$ rankers $\ranker_0$, ..., $\ranker_{K-1}$ are compared in a producer-side experiment where one of them satisfies $\ranker_k=\ranker_{\pureutil}$ as defined in Lemma \ref{lem:util_attn_line_up}. 
If the merged ranker $\ranker$ from the design adheres to both consistency and monotonicity principles, then
	\begin{equation*}
		\forall k'\neq k, \overline{\agg_k}(\util) \geq \overline{\agg_{k'}}(\util),
	\end{equation*}
which implies that the best ranker $\ranker_k$ can be correctly identified based on the expected values of the observed experiment readouts.
\end{corollary}

\section{Solution for the Counterfactual Interleaving Design} \label{sec:new_principles}

In Sections \ref{subsec:facebook}, we have discussed how the existing counterfactual interleaving designs lack a systematic strategy to resolve the merging conflicts when trying to create a valid merged ranker $\ranker$ based on the SUTVA ranker $\ranker_*$. 
Based on the proposed design principles from Section~\ref{sec:math_framework}, we are now able to develop a rigorous solution of counterfactual interleaving designs for producer-side experiments.  

Consider producer-side experiments comparing two ranking models, where all the producer items are split into two groups: the control group $\prodset_0$ and the treatment group $\prodset_1$. We propose to create the counterfactual interleaving design for any possible ramping percentages of producers (Figure~\ref{fig_unicorn}) through the following steps:

Step 1: Generate counterfactual rankings $\ranker_0$ and $\ranker_1$ as if the control or treatment model is ramped to 100\% of the producer items.

Step 2: Merge $\ranker_0$ and $\ranker_1$ to get the SUTVA ranker $\ranker_*$ as defined in equation \eqref{equ:optimal_self_contradicting_design_R}:
For producer items in the control group, get their positions from the control counterfactual ranking $\ranker_0$.
For producer items in the treatment group, get their positions from the treatment counterfactual ranking $\ranker_1$.  

Step 3: Create a valid ranker $\ranker$ based on the SUTVA ranker $\ranker_*$ such that:
\begin{itemize}
\item Preserve the relative order: For any two producer items $d\neq d'$, if $\ranker_*(d) < \ranker_*(d')$, then $\ranker(d) < \ranker(d')$
\item Break the tie according to the probabilistic rule: For any pair of producer items having merging conflicts in $\ranker_*$ (i.e. $d\neq d'$ but $\ranker_*(d) = \ranker_*(d')$), break the tie by placing $d$ before $d'$ in $\ranker$ with probability $\beta_0(d, d')$.
\end{itemize} 

In the above procedure, steps 1 and 2 are the same as the existing solutions while step 3 is a different strategy which is proposed to resolve any possible merging conflicts in the counterfactual interleaving design and ensure an unbiased comparison between treatment and control rankers in the producer-side experiments.
The key is to rigorously preserve the relative order and employ a non-constant tie-breaking probability $\beta_0(d,d')$ whose value we will derive based on the principle of consistency and monotonicity next.

For each spot $j$, consider a pair of producer items $x$ and $y$ satisfying $\ranker_0(x) = \ranker_1(y) = j$. As discussed at the end of Section~\ref{sec:sutva}, $x$ and $y$ would have merging conflict $\ranker_*(x) = \ranker_*(y)$ if and only if $x\in\prodset_0$ and $y\in\prodset_1$. When they have merging conflict, the tie-breaking probability for placing $x$ before $y$ in the merged ranker $\ranker$ can be defined as:
\begin{align*}
\beta_0(x,y) \define \prob{\ranker(x) < \ranker(y)\vert \ranker_*(x) = \ranker_*(y) } = \prob{\ranker(x) < \ranker(y)\vert x\in\prodset_0, y\in\prodset_1}.
\end{align*}
Note that although $\ranker_0(x) = \ranker_1(y) = j$, $\ranker_1(x)$ and $\ranker_0(y)$ can still be larger or smaller than $j$. In the next theorem, we will show that their values should determine the probability $\beta_0(x,y)$.
\begin{theorem} \label{thm:beta0}
    The following $\beta_0(x,y)$ ensures that the merged ranker $\ranker$ satisfies the Principle of Consistency (i.e., $\ranker(x)\vert_{x\in\prodset_0}$ and $\ranker(y)\vert_{y\in\prodset_1}$ have identical distributions):
    \begin{equation}
  \label{equ:beta_0_LHS_RHS_equal}
  \beta_0(x,y) = 
    \begin{cases}
	    p_1\ \textsf{if}\ \ranker_1(x) > j \textsf{ and } \ranker_0(y) > j, \\
	    p_0\ \textsf{if}\ \ranker_1(x) < j \textsf{ and } \ranker_0(y) < j, \\
	    1\ \textsf{if}\ \ranker_1(x) > j \textsf{ and } \ranker_0(y) < j, \\
	    0\ \textsf{if}\ \ranker_1(x) < j \textsf{ and } \ranker_0(y) > j, 
    \end{cases}
\end{equation}
where $p_0$ represents the \% of traffic allocated to the control group and $p_1$ represents the \% of traffic allocated to the treatment group.
\end{theorem}
Proof of this theorem is given in the Appendix~\ref{app:proof_thm_beta0}. In the following theorem, we further prove that the above solution of counterfactual interleaving design also satisfies the monotonicity principle. 
\begin{theorem}
  \label{thm:ntreat1_unicorn_monotonic}
  The merged ranker $\ranker$ in the counterfactual interleaving design created by the above procedure with the tie-breaking probability $\beta_0(d,d')$ derived in equation~\eqref{equ:beta_0_LHS_RHS_equal} at each spot $j$ is both consistent and monotonic.
\end{theorem}

Proof of this theorem is left in the Appendix \ref{app:proof_theorem_ntreat1_unicorn_monotonic}. It is crucial to see that by following the consistency and monotonicity principles, we can obtain a rigorous counterfactual interleaving design solution to ensure valid comparisons between the treatment and conrol rankers in the producer-side experiments.

\section {Examples}\label{sec:exm}
To compare two rankers $R_0$ and $R_1$ in the producer-side experiment,  we have shown how to create a consistent and monotonic merged ranker in the counterfactual interleaving design in Section \ref{sec:new_principles}. In this section, we illustrate the previous theoretical results with both simulated and real examples.

\subsection{Consistent Convolution Kernels and Convoluted Attention Functions}\label{subsec:example_consist}

Consider ten producer items $x_0, \dots, x_9$, and two rankers $\ranker_0$ and $\ranker_1$. Suppose $\ranker_0$ ranks them as:
\begin{equation}
	x_0, x_1, x_2, x_3, x_4, x_5, x_6, x_7, x_8, x_9,
\end{equation}
while the order under $\ranker_1$ is
\begin{equation}
	x_5, x_6, x_7, x_8, x_9, x_0, x_1, x_2, x_3, x_4.
\end{equation}
We first illustrate the convolution kernels and convoluted attention functions as defined in equations \eqref{equ:def_conv_core} and \eqref{equ:def_h_conv}. It is important to note that the convolution kernel $\convcore{0}{j}$ is simply the distribution $\ranker(x_j)$ under the condition $x_j \in \prodset_0$, where $\ranker$ represents the consistent merged ranker and $x_j$ is a producer item such that $\ranker_0(x_j) = j$.

Let us assume that the traffic allocation are $p_0 = 50\%$ to $\prodset_0$ and $p_1 = 50\%$ to $\prodset_1$. In this scenario, $\convcore{0}{j}$ can be computed numerically as depicted in Fig. \ref{fig_conv_core}. Each curve in the chart corresponds to a distribution function $\convcore{0}{j}$ with $j$ indicated in the top-right legend.

\begin{figure}[h]
  \centering
  \includegraphics[width=70mm]{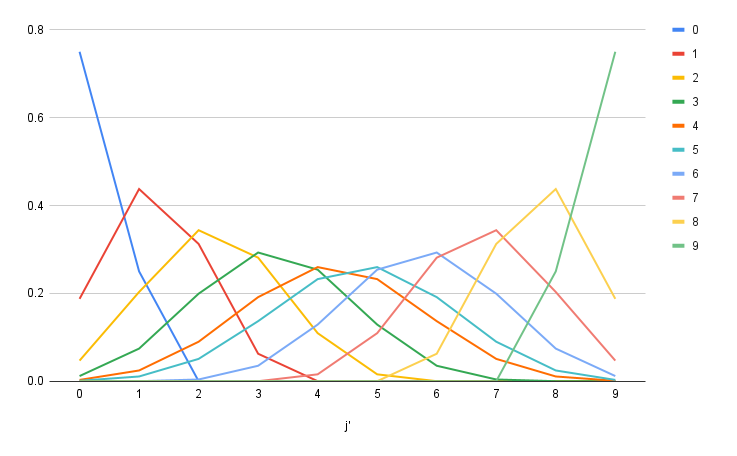}
  \caption{Convolution Kernel Examples}\label{fig_conv_core}
\end{figure}

We can also examine a single convolution kernel, $\convcore{0}{4}$ (i.e., $j = 4$), for different values of traffic allocations $p_0$ as in Fig. \ref{fig_conv_core_j4}.
Each curve in this figure represents the function $\convcore{0}{4}$ for a specific value of $p_0$, as shown in the legend. It is evident that as $p_0$ approaches $1$, the convolution kernel becomes sharper and converges to the delta function, while as $p_0$ approaches $0.5$, the function flattens.

Now, let us verify a series of convoluted attention functions with $p_0 = 0.5$ to confirm that they are indeed monotonic. It can be easily demonstrated that any monotonically decreasing attention function can be decomposed into a sum of functions $\attention_j$, where $\attention_j(j')=\indd{j' \leq j}$. Let the associated convoluted attention function with the convolution kernels depicted in Fig \ref{fig_conv_core} be $\attention^0_j$. These functions are illustrated in Fig \ref{fig_atten_fun_50}. It is evident from the figure that the convoluted attention functions are all monotonically decreasing.

\begin{figure}[h]
\centering
\begin{minipage}{.5\textwidth}
  \centering
  \includegraphics[width=.9\linewidth]{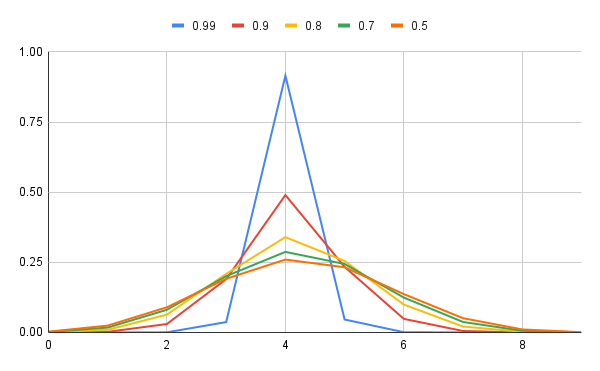}
  \caption{Convolution Kernel Examples}\label{fig_conv_core_j4}
\end{minipage}%
\begin{minipage}{.5\textwidth}
  \centering
  \includegraphics[width=.9\linewidth]{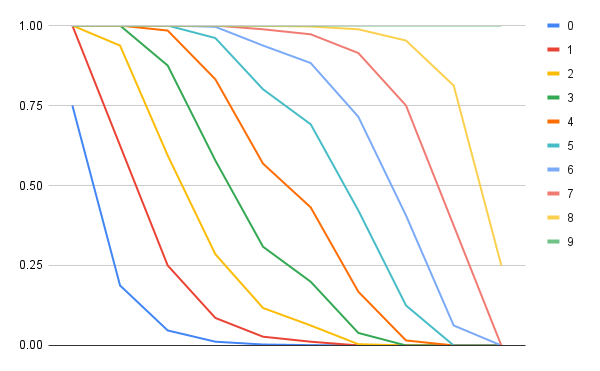}
  \caption{Convoluted Attention Functions}\label{fig_atten_fun_50}
\end{minipage}
\end{figure}

\subsection{Simulation Example}\label{subsec:example_inconsist}

Now we illustrate how the existing counterfactual interleaving designs from Section~\ref{subsec:facebook} can introduce biases in the producer-side experiment readouts and lead to misleading conclusions. For the sake of simplicity, let us consider only four producer items, i.e., $x_0, \dots, x_3$, such that $\ranker_0(x_j) = j$, and their ordering under $\ranker_1$ is $x_1, x_2, x_3, x_0$.
Assume the pure (intrinsic utility) metrics $\pureutil$ for producer item $\left\{ x_i \right\}$ are $\pureutil(x_0) = \pureutil(x_3) = 0.9$ and $\pureutil(x_1) = \pureutil(x_2) = 1$. Clearly, since $\pureutil(x_1) = \pureutil(x_2) > \pureutil(x_3) = \pureutil(x_0)$, the treatment ranker $\ranker_1$ is superior to the control ranker $\ranker_0$.

Existing counterfactual interleaving designs from Section~\ref{subsec:facebook} (such as the \unicorn \ design from \cite{nandy2021unicorn}) employ a naive strategy to resolve merging conflicts  with tie-breaking probability $\beta_0 = 0.5$, which means that if $\ranker(x) = \ranker(y)$ but $x \neq y$, the merged ranker $\ranker$ will place either $x$ or $y$ first with equal probability.
Suppose the attention function $\attention$ has the following form: $\attention(0) = \attention(1) = 1$ and $\attention(2) = \attention(3) = 0$, which means that the top two spots in the recommender system receive full attention from the consumers, whereas the bottom two do not have consumers' attention.

We simulate $N$ independent replications of the experiment, denoted as $\mathcal{S}$. In each replication $s \in \mathcal{S}$, four distinct producer items, $x_{j,s}$, are recommended to the consumer $v_s$, where the pure metric function $\pureutil$, the attention function $\attention$, as well as the rankers $\ranker_0$ and $\ranker_1$ (i.e., $\ranker_0(x_{j,s}) = j$), are the same across all replications. However, the randomization to assign each producer item $x_{j,s}$ into the control group $\prodset_0$ or the treatment group $\prodset_1$ with probability $p_0 = \prob{x_{j,s} \in \prodset_0}$, as well as the randomization to resolve merging conflicts with tie-breaking probability $\beta_0$, are all independent across different replications. 
For any $s \in \mathcal{S}$, the merged ranker $\ranker$ defines the final position of each $x_{j,s}$. 

For each $s \in \mathcal{S}$, we can calculate the experiment readout for the control and treatment groups ($k = 0, 1$) under the merged ranker $\ranker$ using equation~(\ref{equ:aggregate_metrics}):\begin{equation}\label{equ:uks}
	\agg_{k,s}(\util) = \frac{1}{p_k}\sum_{j}\indd{x_{k,s}\in\prodset_k}\pureutil(x_{j,s})\attention\circ\ranker(x_{j,s}),
\end{equation}
where $1/p_k$ is the normalization factor to account for the different sizes of treatment groups (i.e., $p_k$ represents the \% of traffic allocated to the treatment group $k$). Finally, we further take averages of the experiment readouts across $N$ replications for each $k = 0, 1$:
\begin{equation}
	\overline{\agg_k} = \frac{1}{N}\sum_{s\in\mcs} \agg_{k,s}(\util).
\end{equation}
We can also estimate the variance of $\agg_{k,s}(\util)$ by:
\begin{equation}\label{equ:uks_var_est}
	V(\agg_{k,s}(\util)) = \frac{1}{N-1}\left( \sum_{s\in\mcs} \agg^2_{k,s}(\util) - N\overline{\agg}_k^2 \right),
\end{equation}
and estimate the standard deviation of $\overline{\agg_k}$ as:
\begin{equation}
	SD_k = \frac{1}{\sqrt{N}}V(\agg_{k,s}(\util))^{1/2}.
\end{equation}

Under this simulation setup, we expect the correct result as $\overline{\agg_1} > \overline{\agg_0}$, which indicates that $\ranker_1$ is a better ranker than $\ranker_0$.
In the following four simulated cases, Case 1 and 2 are based on the naive \unicorn \ design approach from Section~\ref{subsec:facebook}, which lead to wrong conclusions; Case 3 and 4 leverage the proposed counterfactual interleaving design with consistent merged ranker $\ranker$ from Section~\ref{sec:new_principles}, which could give the correct conclusions.

\textbf{Case 1:} When the traffic allocation to the control group $\prodset_0$ is $p_0 = 0.9$ and a naïve tie-breaker with $\beta_0 = 0.5$ is employed as in the existing \unicorn \ design (Section~\ref{subsec:facebook}), the simulation results for $N = 10^3, 10^4, 10^5$ are as follows:

\begin{center}
\begin{tabular}{ |c|c|c|c| }
	\hline
	$N$ & $10^3$ & $10^4$ & $10^5$ \\ 
	\hline
	$\overline{\agg_0}$ & $1.9609$  & $1.9447$  & $1.9502$     \\  
	\hline
	$\overline{\agg_1}$ & $1.443$  & $1.5973$  & $1.5438$      \\  
	\hline
	$SD_0$ & $0.0126$  & $0.0041$  & $0.0013$      \\  
	\hline
	$SD_1$ & $0.1133$  & $0.0372$  & $0.012$ \\      
	\hline
\end{tabular}
\end{center}

It shows that $\overline{\agg_0} > \overline{\agg_1}$ with high statistical significance, which is misleading. As we have explained earlier, $\ranker_1$ is actually superior to $\ranker_0$.

In this case, the convoluted functions $\attention^0$ and $\attention^1$ can be represented in the following table:
\begin{center}
\begin{tabular}{ |c|c|c|c|c| }
	\hline
	$j$ & $0$ & $1$ & $2$ & $3$  \\ 
	\hline
	$\attention^0(j)$ & $1$ & $0.955$ & $0.095$ & $0$    \\  
	\hline
	$\attention^1(j)$ & $1$ & $0.505$ & $0.045$ & $0$    \\  
	\hline
\end{tabular}
\end{center}
Using these values, we can also directly compute the expected value of the aggregated experiment readouts $\overline{\agg_k}(\util)$ as defined in equation \eqref{equ:agg_k_bar}. By Theorem \ref{thm:conv_attn}, for $k=0, 1$
\begin{equation*}
	\overline{\agg_k}(\util) = \sum_{j=0}^3 \attention^k(j)\pureutil(x_j),
\end{equation*}
and $\overline{\agg_0}(\util) = 1.95$, $\overline{\agg_1}(\util) = 1.5455$. Such results clearly show that, due to the biases introduced by the inconsistent merged ranker, the aggregated experiment readouts incorrectly favor $\ranker_0$ over the actually superior $\ranker_1$.

One might speculate that such bias would not occur if the traffic was equally distributed between the control and treatment groups, i.e., $p_0 = 50\%$, $p_1 = 50\%$. However, in the next case below, we will show that even in this symmetric scenario, the bias persists.

\textbf{Case 2:} Consider the traffic allocation to the control group $\prodset_0$ is $p_0 = 0.5$ and a naïve tie-breaker with $\beta_0 = 0.5$ is employed as in the \unicorn \ approach (Section~\ref{subsec:facebook}). The simulation results are:

\begin{center}
\begin{tabular}{ |c|c|c|c| }
	\hline
	$N$ & $10^3$ & $10^4$ & $10^5$ \\ 
	\hline
	$\overline{\agg_0}$ & $2.081$  & $2.139$  & $2.155$     \\  
	\hline
	$\overline{\agg_1}$ & $1.808$  & $1.749$  & $1.733$      \\  
	\hline
	$SD_0$ & $0.04$  & $0.013$  & $0.004$      \\  
	\hline
	$SD_1$ & $0.04$  & $0.013$  & $0.004$ \\      
	\hline
\end{tabular}
\end{center}
In this csae, the convoluted attention functions can be calculated as follows:
\begin{center}
\begin{tabular}{ |c|c|c|c|c| }
	\hline
	$j$ & $0$ & $1$ & $2$ & $3$  \\ 
	\hline
	$\attention^0(j)$ & $1$ & $0.875$ & $0.375$ & $0$    \\  
	\hline
	$\attention^1(j)$ & $1$ & $0.625$ & $0.125$ & $0$    \\  
	\hline
\end{tabular}
\end{center}
and consequently, $\overline{\agg_0}(\util) = 2.15$, $\overline{\agg_1}(\util) = 1.7375$. Once again, the bias is clearly evident and the experiment readouts incorrectly indicate that $\ranker_0$ is superior to $\ranker_1$ with high significance..

In order to achieve a valid comparison between the treatment and control rankers, we should use the proposed counterfactual interleaving design from Section~\ref{sec:new_principles}, which is guaranteed to yield a consistent merged ranker. In the next two cases, we will show that the proposed consistent ranker can effectively identify the superior ranker in the simulated producer-side experiment.

\textbf{Case 3:} When the traffic allocation to the control group $\prodset_0$ is $p_0 = 0.9$ and the proposed counterfactual interleaving design with consistent merged ranker $\ranker$ from Section~\ref{sec:new_principles} is used, we have

\begin{center}
\begin{tabular}{ |c|c|c|c| }
	\hline
	$N$ & $10^3$ & $10^4$ & $10^5$  \\
	\hline
	$\overline{\agg_0}$ & $1.880$  & $1.893$  & $1.900$     \\  
	\hline
	$\overline{\agg_1}$ & $2.161$  & $2.06$  & $2.000$      \\  
	\hline
	$SD_0$ & $0.015$  & $0.0045$  & $0.0014$      \\  
	\hline
	$SD_1$ & $0.13$  & $0.041$  & $0.013$ \\      
	\hline
\end{tabular}
\end{center}
Based on the consistent ranker, we can clearly see that $\overline{\agg_1} > \overline{\agg_0}$ and $\ranker_1$ is identified as the optimal ranker.

\textbf{Case 4:} Consider the traffic allocation to the control group $\prodset_0$ is $p_0 = 0.5$ and the proposed counterfactual interleaving design with consistent merged ranker $\ranker$ from Section~\ref{sec:new_principles} is used, we have:

\begin{center}
\begin{tabular}{ |c|c|c|c| }
	\hline
	$N$ & $10^3$ & $10^4$ & $10^5$ \\ 
	\hline
	$\overline{\agg_0}$ & $1.956$  & $1.888$  & $1.904$     \\  
	\hline
	$\overline{\agg_1}$ & $1.920$  & $1.987$  & $1.971$      \\  
	\hline
	$SD_0$ & $0.04$  & $0.012$  & $0.004$      \\  
	\hline
	$SD_1$ & $0.04$  & $0.012$  & $0.004$ \\      
	\hline
\end{tabular}
\end{center}
Once again, different from the existing \unicorn \ design in Section~\ref{subsec:facebook}, the proposed solution from Section~\ref{sec:new_principles} can draw the correct conclusion.

\subsection{Recommender System Example from Online Social Networks} \label{sec:PYMKexample}

Online social network platforms play a crucial role in connecting people with one another, offering features such as Feeds and People You May Know (PYMK). A key aspect of these platforms is their ability to recommend a ranked list of users/creators for viewers to follow or connect with. Such recommender systems are vital in shaping the experience of both viewers and creators on the social network, and any new changes to the recommender's ranking algorithm need to be carefully evaluated through online experiments before getting fully deployed in production.
In this context, viewers are the consumers who consume content on the network, and standard viewer-side A/B tests can be used to assess the impact of ranking changes on the viewers' behavior. On the other hand, producers on the network are the creators who are ranked and recommended by the platform for viewers to connect or follow. Measuring the impact of ranking changes on the creators through producer-side experiment is equally important for improving the ecosystem of online social networks. Such producer-side experiments are also often referred to as the \textit{creator-side experiments}. 

The AI team at LinkedIn initially implemented the UniCoRn design (Section~\ref{subsec:facebook}) for running creator-side experiments in the online edge recommender system, serving tens of millions of members and providing billions of edge recommendations daily. However, it became evident that the UniCoRn-based approach led to biased creator positions in the final ranking, and readouts from the corresponding creator-side experiments were difficult to interpret.
After recognizing the importance of consistency and monotonicity principles in designing creator-side experiments, the team implemented the new counterfactual interleaving design as proposed in Section~\ref{sec:new_principles}. 

\begin{figure}[h]
  \centering
  \includegraphics[width=160mm]{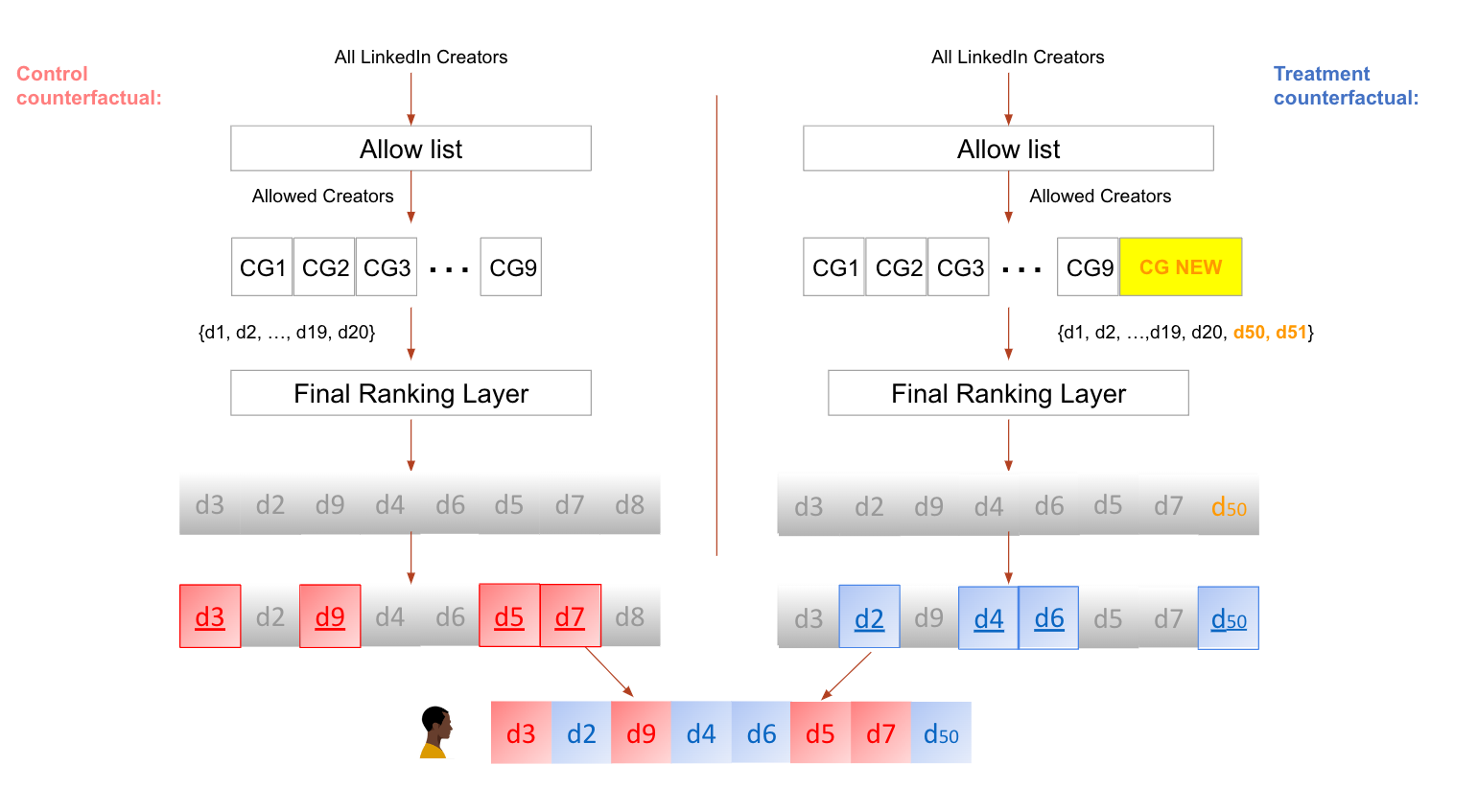}
  \caption{Illustration of a Counterfactual Interleaving Design for LinkedIn's Online Edge Recommender System}\label{fig_example}
\end{figure}

As shown in Figure~\ref{fig_example}, the ranker in the online edge recommender system at LinkedIn consists of three sequential ranking/filtering layers with increasing complexities: (1) \textit{Allow List} is a simple rule-based filtering layer, which returns a subset of creators who are eligible to be recommended to the viewer. (2) \textit{Candidate Generator} (CG) layer scores and ranks all creators in the allow list, and returns the top $m$ candidates. There can be multiple independent CGs in this layer (e.g., one CG for each country or industry segment) and union of all the selected candidates will be sent to the final ranking layer. The ranking models in the CG layer are generally easier to compute and thus they can be used to score and rank a large number of creators in the allow list. (3) \textit{Final Ranking} layer scores and ranks all the candidates selected by CG layer using more sophisticated models. In the end, top $n$ ($n < m$) creators out of the final ranking is shown to the viewer. Figure~\ref{fig_example} illustrates the counterfactual interleaving design of a creator-side experiment at LinkedIn where the treatment is to add a new CG to the existing CGs (e.g., CG1, CG2, …, CG9). 

Generating the treatment and control counterfactual rankings for the counterfactual interleaving design may not always require running the Allow List, CG and Final Ranking layers twice. In some cases, it is possible to develop computational shortcuts. Take the counterfactual interleaving design in Figure~\ref{fig_example} for example. 
Because the candidates generated in the control CG is a subset of those candidates generated in the treatment CG (while Allow List and Final Ranking layers are the same between treatment and control), the creator-side experiment only needs to run Allow List, CG and Final Scoring/Ranking for the treatment counterfactual case. Then, as shown in Figure~\ref{fig_shortcut}, the control counterfactual ranking can be directly obtained based on the treatment counterfactual ranking (after removing any candidates that were in the treatment CG but not in the control CG). In other words, the only extra computation needed for generating the control counterfactual ranking is to read a few more creators from the treatment counterfactual ranking list to fill the empty spots at the end of the control counterfactual ranking list. This shortcut can substantially reduce the computation especially when scoring and ranking a large number of creators are expensive. 

\begin{figure}[h]
  \centering
  \includegraphics[width=160mm]{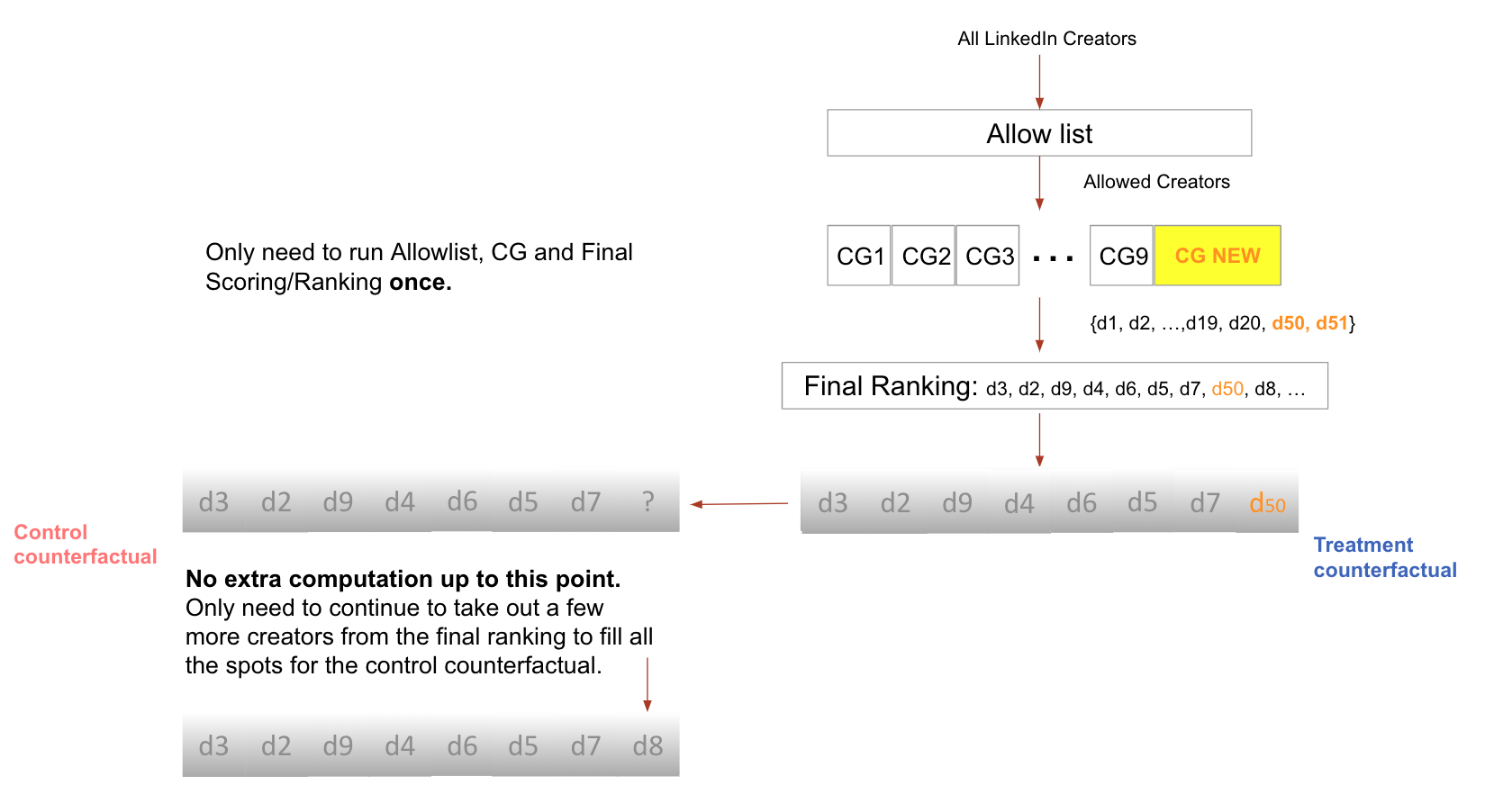}
  \caption{Illustration of the Computational Shortcut for Generating the Treatment and Control Counterfactual Rankings}\label{fig_shortcut}
\end{figure}

After deploying the proposed solution in production, a comparison between the new counterfactual interleaving design and the previous UniCoRn design has revealed that the UniCoRn-based design had resulted in approximately 85\% of creators being placed in the wrong positions. On average, every creator was randomly shifted away by $\pm$ 2 or 3 positions in a recommendation session compared to their correct positions based on the counterfactual rankings. 
Using the proposed new counterfactual interleaving design, the team was able to obtain fair comparisons between the treatment and control rankers, ensuring trustworthy creator-side experiment readouts. 
By running both viewer-side (consumer-side) and creator-side (producer-side) experiments and evaluating the ranking changes' impacts on both sides of the marketplace, the team can strike a balance that benefits all stakeholders involved. As recommender systems continue to evolve, these considerations will play an increasingly pivotal role in enhancing user experiences and driving success in online platforms.

\section{Conclusions} \label{sec:conclusions}

Many online platforms are two-sided marketplaces which have producers (e.g., sellers, content creators, hosts) on one side and consumers (e.g., buyers, content viewers, and customers) on the other side. Recommender systems aim to predict consumer preferences and allocate more preferred producer items to spots where consumers are likely to pay greater attention. To optimize an online recommender system, it is critical to conduct online experiments to thoroughly evaluate the impacts of any new ranking model changes on both sides. While consumer-side impact can be easily measured via simple online A/B testing, producer-side measurement is much more challenging. 
In this paper, we scrutinize issues of the current ad hoc design solutions in the literature and propose general principles for designing trustworthy online producer-side experiments. Building upon the proposed consistency and monotonicity principles,  we also derive a rigorous counterfactual interleaving design solution to ensure valid comparisons between treatment and control rankers. The proposed methodology and design principles can serve as guidelines for online platforms seeking to improve their recommender systems and ensuring accuracy in their evaluations on the producer side. 

In the end, we also want to note that an alternative way to measure producer-side impacts is through cluster-randomized experiments \citep{karrer_cluster_2021,  saveski_cluster_2017, saintjacques_egocluster2019}, where consumers are partitioned into various disjoint clusters and each cluster is associated with one producer. Such solution for measuring producer-side effects has two major limitations: (1) the effective sample size (and hence the power) of cluster-randomized experiments tends to be small; and (2) it is often challenging to partition the network into clusters and different clustering algorithms can lead to different experiment results. Moreover, in some applications (such as the online edge recommender system described in Section~\ref{sec:PYMKexample}), it is not possible to run cluster-randomized experiments because the new treatment in the experiment would keep changing the edge structure of the online social network.  

\section{Acknowledgements}

The authors would like to thank Nian Si, Preetam Nandy, Weitao Duan, James Sorenson, Cindy Liang, Parag Agrawal, Andrew Hatch, Chun Lo, Yafei Wei, Liyan Fang, Wentao Su and Wanjun Liu for their suggestions and feedbacks. The authors also would like to thank the researchers and engineers from the Data Science Applied Research team, Follows AI team and PYMK AI team at LinkedIn.

\appendix

\section{Proof of Theorem \ref{thm:conv_attn}}\label{app:proof_thm_conv_attn}

We first calculate the following sum of conditional expectations:

\begin{align} \label{equ:proof_thm_conv_attn_part0}
  \begin{split}
     & \sum_{d\in \prodset }\expected[\einfo]{\util\left(d; \ranker\right)\vert d\in\prodset_k} \\
     = & \sum_{j\in\spotset}\expected[\einfo]{\util\left(\spotassigner_k(j); \ranker\right)\vert \spotassigner_k(j)\in\prodset_k} \\
     = & \sum_{j\in\spotset}\pureutil\circ\spotassigner_k(j)\times\expected[\einfo]{\attention\circ\ranker\circ\spotassigner_k(j)\vert \spotassigner_k(j)\in\prodset_k} \\
     = & \sum_{j\in\spotset}\pureutil\circ\spotassigner_k(j)\times \attention^k(j)
  \end{split}
\end{align}

Now we prove Theorem \ref{thm:conv_attn} by calculating $p_k\overline{\agg_k}(\util)$: 
\begin{align} \label{equ:proof_thm_conv_attn_part1}
  \begin{split}
    & p_k\overline{\agg_k}(\util) \\
    = & \expected{\sum_{d\in \prodset_k} \util\left(d; \ranker\right)\bigg|\finfo} \\
    = & \expected{\sum_{d\in \prodset } \util\left(d; \ranker\right)\indd{d\in\prodset_k}\bigg|\finfo} \\
    = &\sum_{d\in \prodset }\expected{\util\left(d; \ranker\right)\indd{d\in\prodset_k}\bigg|\finfo}\\
    = &\sum_{d\in \prodset }\expected[\einfo]{\util\left(d; \ranker\right)\indd{d\in\prodset_k}}\ \text{(by independence of $\einfo$ and $\finfo$)}\\
    = & p_k\sum_{d\in \prodset }\expected[\einfo]{\util\left(d; \ranker\right)\vert d\in\prodset_k} \\
    = & p_k\sum_{j\in\spotset}\pureutil\circ\spotassigner_k(j)\times \attention^k(j)\ \text{(plugging in equation \eqref{equ:proof_thm_conv_attn_part0})} 
  \end{split}
\end{align}
and the proof is complete.

\section{Proof of Lemma \ref{lem:util_attn_line_up}} \label{app:proof_util_attn_line_up}

The crux of the proof lies in observing that, since the attention function $\attention$ is monotonically decreasing by definition, for any ranker $\ranker \neq \ranker_{\pureutil}$, there must exist producer items $d_0, d_1 \in \prodset$ such that $\pureutil(d_0) > \pureutil(d_1)$ but $\ranker(d_0) > \ranker(d_1)$, which implies that $\attention \circ \ranker(d_0) \leq \attention \circ \ranker(d_1)$. Consider a new ranker, $\tilde{\ranker}$, which is identical to $\ranker$ except that the ranks of $d_0$ and $d_1$ are swapped. Then,

  \begin{align*}
    \begin{split}
    & \sum_{d\in\prodset}\pureutil(d)\times\attention\circ\tilde\ranker(d) - \sum_{d\in\prodset}\pureutil(d)\times\attention\circ\ranker(d) \\
     = & \pureutil(d_0)\times\attention\circ\tilde\ranker(d_0) + \pureutil(d_1)\times\attention\circ\tilde\ranker(d_1) - \\
     & - \pureutil(d_0)\times\attention\circ\ranker(d_0) - \pureutil(d_1)\times\attention\circ\ranker(d_1) \\ 
     = & \pureutil(d_0)\times\attention\circ\ranker(d_1) + \pureutil(d_1)\times\attention\circ\ranker(d_0) - \\
     & - \pureutil(d_0)\times\attention\circ\ranker(d_0) - \pureutil(d_1)\times\attention\circ\ranker(d_1) \\  
     = & \left[ \pureutil(d_0) - \pureutil(d_1) \right]\times\left[ \attention\circ\ranker(d_1) - \attention\circ\ranker(d_0) \right]\\
     \ge & 0
    \end{split}
  \end{align*}

This inequality indicates that swapping the ranks of $d_0$ and $d_1$ leads to a non-decrease in the value of $\sum_{d \in \prodset} \pureutil(d) \times \attention \circ \ranker(d)$. By iteratively swapping such pairs $(d_0, d_1)$, $\ranker$ can eventually be transformed into $\ranker_{\pureutil}$. Throughout this process, the value of $\sum_{d \in \prodset} \pureutil(d) \times \attention \circ \ranker(d)$ never decreases, thereby establishing $\ranker_{\pureutil}$ as the maximizer.

\section{Proof of Theorem~\ref{thm:beta0}} \label{app:proof_thm_beta0}

Proof:
It is evident that $\ranker(x)$ is equal to one plus the number of producer items ranked ahead of $x$ by $\ranker$:

\begin{equation*}
\ranker(x) = 1 + \sum_{d\neq x}\indd{\ranker(d)<\ranker(x)}.
\end{equation*}

Considering that $\ranker$ maintains the relative order of $\ranker_*$, and $d = y$ is the sole producer item (apart from $x$) for which $\ranker_*(d) = \ranker_*(x)$, under the condition $x\in\prodset_0$ (or equivalently, $\ranker_*(x)=j$),

\begin{equation}\label{equ:Rd0Decompose}
	\begin{split}
		\ranker(x)\Big|_{x\in\prodset_0} & = 1 + \sum_{d\not\in\left\{ x, y \right\}}\indd{\ranker(d)<\ranker(x)} + \indd{\ranker(y) < \ranker(x)}\Big|_{x\in\prodset_0} \\
	& = 1 + \sum_{d\not\in\left\{ x, y \right\}}\indd{\ranker_*(d)<j}\Big|_{x\in\prodset_0} + \indd{\ranker(y) < \ranker(x)}\Big|_{x\in\prodset_0} \\
	& = 1 + O^0_j + Q^0_j,
	\end{split}
\end{equation}

where
\begin{equation}\label{equ:XjDef}
	O_j\define\sum_{d\not\in\left\{ x, y \right\}}\indd{\ranker_*(d)<j}
\end{equation}
and $O^0_j = O_j\vert_{x\in\prodset_0}$. Additionally 
\begin{equation}\label{equ:delta0jDef}
	Q^0_j\define\indd{\ranker(y) < \ranker(x)}\vert_{x\in\prodset_0}.
\end{equation}

Likewise, under the condition $y\in\prodset_1$,

\begin{equation}\label{equ:Rd1Decompose}
	\ranker(y)\Big|_{y\in\prodset_1}  = 1 + O^1_j + Q^1_j
\end{equation}
with $O^1_j = O_j\vert_{y\in\prodset_1}$ and $Q^1_j = \indd{\ranker(x) < \ranker(y)}\vert_{y\in\prodset_1}$. 

The variation in the variable $O_j$ is influenced by the treatment allocations of producer items excluding $x$ and $y$. On the other hand, the variation in the terms $Q^0_j$ and $Q^1_j$ is governed by the allocations of $x$, $y$, as well as the random generator, $\randgen$, responsible for tie-breakings. Therefore, for $k=0, 1$, $O^k_j$ and $Q^k_j$ are independent. Furthermore, $O^0_j$ and $O^1_j$ possess identical distributions, as both are identical to the distribution of $O_j$.

For $\ranker(x)\vert_{x\in\prodset_0}$ and $\ranker(y)\vert_{y\in\prodset_1}$ to exhibit the same distribution (in order to satisfy the Principle of Consistency), we need to ensure that $Q^0_j$ and $Q^1_j$ have identical distributions. Based on the definitions of each term, this implies
\begin{equation}\label{equ:delta01j_identically_distributed}
	\prob{\ranker(y) < \ranker(x)\vert x\in\prodset_0} = \prob{\ranker(x) < \ranker(y)\vert y\in\prodset_1}.  
\end{equation}
For the left-hand side:
\begin{align*}
	\begin{split}
	& \prob{\ranker(y) < \ranker(x)\vert x\in\prodset_0} \\ 
	 =\ & p_0\prob{\ranker(y) < \ranker(x)\vert x\in\prodset_0, y\in\prodset_0} \\ 
	& + p_1\prob{\ranker(y) < \ranker(x)\vert x\in\prodset_0, y\in\prodset_1} \\
	 =\ & p_0\indd{\ranker_0(y)<j} + p_1(1-\beta_0),
	\end{split}
\end{align*}
and similarly for the right-hand side:
\begin{align*}
	 \prob{\ranker(x) < \ranker(y)\vert y\in\prodset_1}  =  p_1\indd{\ranker_1(x)<j} + p_0\beta_0,
\end{align*}
where $p_0 = \frac{|\prodset_0|}{|\prodset_0|+|\prodset_1|}$ representing the \% of traffic allocated to the control group and $p_1=\frac{|\prodset_1|}{|\prodset_0|+|\prodset_1|}$ representing the \% of traffic allocated to the treatment group.

By equalizing the above two sides, we can obtain $\beta_0$ under the Principle of Consistency as:

\begin{equation}
  \label{equ:beta_0_LHS_RHS_equal_app}
  \beta_0 = 
    \begin{cases}
	    p_1\ \textsf{if}\ \ranker_1(x) > j \textsf{ and } \ranker_0(y) > j, \\
	    p_0\ \textsf{if}\ \ranker_1(x) < j \textsf{ and } \ranker_0(y) < j, \\
	    1\ \textsf{if}\ \ranker_1(x) > j \textsf{ and } \ranker_0(y) < j, \\
	    0\ \textsf{if}\ \ranker_1(x) < j \textsf{ and } \ranker_0(y) > j. 
    \end{cases}
\end{equation}

This definition ensures that both $Q^0_j$ and $Q^1_j$ follow Bernoulli distributions with identical expected values as shown below:

\begin{equation}
  \label{equ:expected_value_delta_j}
  \expected{Q^0_j} = 
    \begin{cases}
	    p_0p_1\ \textsf{if}\ \ranker_1(x) > j \textsf{ and } \ranker_0(y) > j, \\
	    1 - p_0p_1\ \textsf{if}\ \ranker_1(x) < j \textsf{ and } \ranker_0(y) < j, \\
	    p_0\ \textsf{if}\ \ranker_1(x) > j \textsf{ and } \ranker_0(y) < j, \\
	    p_1\ \textsf{if}\ \ranker_1(x) < j \textsf{ and } \ranker_0(y) > j. 
    \end{cases}
\end{equation}

\section{Proof of Theorem \ref{thm:ntreat1_unicorn_monotonic}}
\label{app:proof_theorem_ntreat1_unicorn_monotonic}
We start the proof of Theorem $\ref{thm:ntreat1_unicorn_monotonic}$ with the following straightforward yet essential results.

\begin{lemma}
\label{lem:sumOfIndVariables}
Let $X$, $Y_1$, and $Y_2$ be random variables, where $X$ and $Y_i$ are independent for $i = 1, 2$. If $Y_1 \succ Y_2$, then $X + Y_1 \succ X + Y_2$.
\end{lemma}
\begin{proof}
Consider any number $t$. It suffices to show that $\prob{X + Y_1 \geq t} \geq \prob{X + Y_2 \geq t}$.
\begin{equation*}
	\begin{split}
		\prob{X + Y_1 \geq t} & = \expected{\prob{Y_1 \geq t-x}\vert_{x = X}} \\
		& \geq \expected{\prob{Y_2 \geq t-x}\vert_{x=X}} \\
		& = \prob{X + Y_2 \geq t}
	\end{split}
\end{equation*}
\end{proof}

\begin{corollary}\label{cor:sumOfIndVariables}
Let $X_1$, $X_2$, $Y_1$, and $Y_2$ be random variables, where $X_i$ and $Y_i$ are independent for $i = 1, 2$. If $X_1 \succ X_2$ and $Y_1 \succ Y_2$, then $X_1 + Y_1 \succ X_2 + Y_2$.
\end{corollary}
\begin{proof}
Leveraging Lemma \ref{lem:sumOfIndVariables}, we can deduce that $X_1 + Y_1 \succ X_1 + Y_2$ and $X_1 + Y_2 \succ X_2 + Y_2$. Consequently, $X_1 + Y_1 \succ X_2 + Y_2$.
\end{proof}

To streamline the discussion, let's introduce a set of notations. Let $X$ and $Y$ be two random variables, which might not necessarily be independent. Define $X\oplus Y$ as a random variable $Z$ such that $Z = X' + Y'$, where $X'$ and $Y'$ are identically distributed as $X$ and $Y$, respectively, and are also independent.

With Corollary \ref{cor:sumOfIndVariables} in mind, we can deduce the following corollary:

\begin{corollary}\label{cor:directSumOfRV}
Let $X_1$, $X_2$, $Y_1$, and $Y_2$ be random variables. If $X_1 \succ X_2$ and $Y_1 \succ Y_2$, then $X_1 \oplus Y_1 \succ X_2 + Y_2$.
\end{corollary}

For $0 \leq p \leq 1$, let $\gamma(p)$ denote a random variable following a Bernoulli distribution with mean $p$. Specifically, $\gamma(0) = 0$ and $\gamma(1) = 1$. The following lemma is self-evident:

\begin{lemma}
	\label{lem:bernoulliOrder}
  Let $0\leq p \leq q \leq 1$, then $\gamma(q)\succ \gamma(p)$.
\end{lemma}

For any position $j$, let us denote by $x_j$, $y_j$ the producer items such that $\ranker_0(x_j) = \ranker_1(y_j) = j$. According to equation \eqref{equ:def_conv_core}, the convolution kernel $\convcore{k}{j}$ is essentially the distribution of $\ranker(x_j)$ when $x_j\in\prodset_0$ or $\ranker(y_j)$ when $y_j\in\prodset_1$. Both distributions are identical when $\ranker$ is the consistent merged ranker, and the tie-breakers are computed as described in Section \ref{sec:new_principles}.

It is noteworthy that $x_j$ and $y_j$ may be equal for certain $j$. In such cases, no tie-breaker is needed at that position because no conflicts arise. Let's define $\spotpotcands(j)$ as the set containing $x_j$ and $y_j$. When $x_j = y_j$, $\spotpotcands(j)$ will contain only one element.

By definition, $\ranker(x_j) = 1 + \sum_{d\in\prodset} \indd{\ranker(d) < \ranker(x_j)}$, which leads us to

\begin{equation}\label{equ:Rd0DecomposeV2}
	\begin{split}
		\ranker(x_j)\Big|_{x_j\in\prodset_0} & = 1 + \sum_{d\not\in\spotpotcands(j)}\indd{\ranker(d)<\ranker(x_j)} + \sum_{d\in\spotpotcands(j)}\indd{\ranker(d)<\ranker(x_j)}\Big|_{x_j\in\prodset_0} \\
	& = 1 + \sum_{d\not\in\spotpotcands(j)}\indd{\ranker_*(d)<j}\Big|_{x_j\in\prodset_0} + \sum_{d\in\spotpotcands(j)}\indd{\ranker(d)<\ranker(x_j)}\Big|_{x_j\in\prodset_0} \\
	& \define 1 + O^0_j + Q^0_j
	\end{split}
\end{equation}

Equation \eqref{equ:Rd0DecomposeV2} is consistent with equation \eqref{equ:Rd0Decompose}, but accounts for the case where $\sizeof{\spotpotcands(j)} = 1$.

As outlined in Appendix \ref{app:proof_thm_beta0}, the two terms $O^0_j$ and $Q^0_j$ are independent, and $O^0_j$ is distributed the same as

\begin{equation}\label{equ:XjDecompose2kdfjd}
	O_j = \sum_{d\not\in\spotpotcands(j)}\indd{\ranker_*(d)<j}.
\end{equation}

$O_{j}$ can be further decomposed into two independent components, both of which are also independent of $Q^0_j$:

\begin{equation}\label{equ:Oj}
  \begin{split}
	O_j & = \sum_{d\not\in\spotpotcands(j)\cup\spotpotcands(j+1)}\indd{\ranker_*(d)<j} + \sum_{d\in\spotpotcands(j+1)\setminus\spotpotcands(j)}\indd{\ranker_*(d)<j} \\
	& \define \hat{O}_j + O'_{j}.
  \end{split}
\end{equation}

Similarly for $O_{j+1}$:
\begin{equation}\label{equ:Ojplus1}
  \begin{split}
	O_{j+1} & = \sum_{d\not\in\spotpotcands(j)\cup\spotpotcands(j+1)}\indd{\ranker_*(d)<j+1} + \sum_{d\in\spotpotcands(j)\setminus\spotpotcands(j+1)}\indd{\ranker_*(d)<j+1} \\
	& \define \hat{O}_{j+1} + O''_{j+1}.
  \end{split}
\end{equation}

By definition, when $d$ is not an element of the set $\spotpotcands(j) \cup \spotpotcands(j+1)$, the indicator function $\indd{\ranker_*(d)<j}$ is equal to $\indd{\ranker_*(d)<j+1}$. As a result, $\hat{O}_j$ is equal to $\hat{O}_{j+1}$.

Let's introduce $T(j) = O'_j \oplus Q^0_j$ and $T(j+1) = O''_{j+1} \oplus Q^0_{j+1}$. According to Corollary \ref{cor:directSumOfRV}, in order to prove Theorem \ref{thm:ntreat1_unicorn_monotonic}, it is sufficient to demonstrate that for any integer $j$ greater than or equal to $1$, $T(j+1)$ stochastic dominates $T(j)$, that is, $T(j+1)\succ T(j)$.

Regarding $Q^0_j$, if the size of the set $\spotpotcands(j)$ is equal to $1$, then $Q^0_j$ is equal to $0$. On the other hand, if the size of the set $\spotpotcands(j)$ is equal to $2$, $Q^0_j$ follows a Bernoulli distribution with an expected value that can be computed using equation \eqref{equ:expected_value_delta_j}, by substituting $x$ and $y$ with $x_j$ and $y_j$, respectively.

For the sake of simplification, let's denote $\ixj$ as $\indd{\ranker_1(x_j)<j}$ and $\iyj$ as $\indd{\ranker_0(y_j)<j}$. It is important to note that both $\ixj$ and $\iyj$ are deterministic functions that can only take the values $0$ or $1$.

\begin{lemma} \label{lem:QjCj2}
For any $j$, if $\sizeof{\spotpotcands(j)} = 2$, then the distribution of $Q^0_j$ is determined by the values of $\ixj$ and $\iyj$ as below:
\begin{center}
\begin{tabular}{ |c|c|c| }
	\hline
	$\ixj$ & $\iyj$ & $Q^0_j$  \\ 
	\hline
 0 & 0 & $\gamma(p_0p_1)$    \\  
	\hline
 1 & 1 & $\gamma(1 - p_0p_1)$    \\  
	\hline
 0 & 1 & $\gamma(p_0)$    \\  
	\hline
 1 & 0 & $\gamma(p_1)$    \\  
	\hline
\end{tabular}
\end{center}
\end{lemma}
\begin{proof}
When $\sizeof{\spotpotcands(j)} = 2$, $x_j\neq y_j$, so 
\begin{equation*}
	Q^0_j = \indd{\ranker(y_j) < \ranker(x_j)}\Big|_{x_j\in \prodset_0},
\end{equation*}
and the conclusion follows from equation \eqref{equ:expected_value_delta_j}.
\end{proof}

\begin{lemma}
	\label{lem:OjCj1_2}
For any $j$, if $\sizeof{\spotpotcands(j+1)} = 2$ and $\spotpotcands(j+1)\cap\spotpotcands(j)=\varnothing$, then the distribution of $O'_j$ is determined by the values of $\ixjp$ and $\iyjp$ as below:
\begin{center}
\begin{tabular}{ |c|c|c| }
	\hline
	$\ixjp$ & $\iyjp$ & $O'_j$  \\ 
	\hline
 0 & 0 & $0$    \\  
	\hline
 1 & 1 & $\gamma(p_0)\oplus\gamma(p_1)$    \\  
	\hline
 0 & 1 & $\gamma(p_0)$    \\  
	\hline
 1 & 0 & $\gamma(p_1)$    \\  
	\hline
\end{tabular}
\end{center}
\end{lemma}
\begin{proof}
	By equation \eqref{equ:Oj}, 

\begin{equation*}
  \begin{split}
	O'_{j} & = \sum_{d\in\spotpotcands(j+1)\setminus\spotpotcands(j)}\indd{\ranker_*(d)<j} \\
	& = \sum_{d\in\spotpotcands(j+1)}\indd{\ranker_*(d)<j} \\
	& = \indd{\ranker_*(x_{j+1})<j} + \indd{\ranker_*(y_{j+1})<j} \\
	& = \indd{\ranker_*(x_{j+1})<j+1} + \indd{\ranker_*(y_{j+1})<j+1} (\textit{b.c. }\spotpotcands(j+1)\cap\spotpotcands(j)=\varnothing) \\
	& = \ixjp\indd{x_{j+1}\in\prodset_1} + \iyjp\indd{y_{j+1}\in\prodset_0}
  \end{split}
\end{equation*}
\end{proof}

\begin{lemma}
	\label{lem:Oj1Cj_2}
	For any $j$, if $\sizeof{\spotpotcands(j)} = 2$ and $\spotpotcands(j+1)\cap\spotpotcands(j)=\varnothing$, then the distribution of $O''_{j+1}$ is determined by the values of $\ixj$ and $\iyj$ as below:
\begin{center}
\begin{tabular}{ |c|c|c| }
	\hline
	$\ixj$ & $\iyj$ & $O''_{j+1}$  \\ 
	\hline
 0 & 0 & $\gamma(p_0)\oplus\gamma(p_1)$    \\  
	\hline
 1 & 1 & $2$    \\  
	\hline
 0 & 1 & $\gamma(p_0)\oplus \gamma(p_0) \oplus \gamma(p_1)$    \\  
	\hline
 1 & 0 & $\gamma(p_1)\oplus \gamma(p_1) \oplus \gamma(p_0)$    \\  
	\hline
\end{tabular}
\end{center}
\end{lemma}
\begin{proof}
	By equation \eqref{equ:Ojplus1}, 

\begin{equation*}
  \begin{split}
	O''_{j+1} & = \sum_{d\in\spotpotcands(j)\setminus\spotpotcands(j+1)}\indd{\ranker_*(d)<j+1} \\
	& = \sum_{d\in\spotpotcands(j)}\indd{\ranker_*(d)<j+1} \\
	& = \indd{\ranker_*(x_{j})<j+1} + \indd{\ranker_*(y_{j})<j+1} \\
	& = \indd{x_j\in\prodset_0} + \indd{\ranker_1(x_{j}<j+1)}\indd{x_{j}\in\prodset_1} \\ 
	& \ + \indd{y_j\in\prodset_1} + \indd{\ranker_0(y_{j}<j+1)}\indd{y_{j}\in\prodset_0} \\
	& = \indd{x_j\in\prodset_0} + \ixj\indd{x_{j}\in\prodset_1} \\ 
	& \ + \indd{y_j\in\prodset_1} + \iyj\indd{y_{j}\in\prodset_0} (\textit{b.c. }\spotpotcands(j+1)\cap\spotpotcands(j)=\varnothing)
  \end{split}
\end{equation*}
\end{proof}

\begin{lemma}
	\label{lem:CjCj1Size1}
	If $\sizeof{\spotpotcands(j)} = \sizeof{\spotpotcands(j+1)} = 1$, then $T(j+1)\succ T(j)$.
\end{lemma}
\begin{proof}
	In this case $x_j = y_j$, $x_{j+1} = y_{j+1}$ while $x_j \neq x_{j+1}$. By equation \eqref{equ:Rd0DecomposeV2}, $Q^0_j = Q^0_{j+1} = 0$. By equation \eqref{equ:Oj}, $O'_{j} = 0$. By equation \eqref{equ:Ojplus1}, $O''_{j+1} = 1$. Therefore, $T(j+1) = 1$ and $T(j) = 0$.
\end{proof}

\begin{lemma} \label{lem:Cj1Size2Qprime}
For any $j$, if $\sizeof{\spotpotcands(j+1)} = 2$ and $\sizeof{\spotpotcands(j)} = 1$, then $T(j+1)\succ T(j)$.
\end{lemma}
\begin{proof}
	In this case $\spotpotcands(j+1)\cap\spotpotcands(j)=\varnothing$. So $Q^0_j = 0$ by equation \eqref{equ:Rd0DecomposeV2} while $O''_{j+1} = 1$ by equation \eqref{equ:Ojplus1}. Meanwhile the distribution of $Q^0_{j+1}$ and $O'_j$ can be obtained by Lemma \ref{lem:QjCj2} (substitute $j\rightarrow j+1$) and Lemma \ref{lem:OjCj1_2}. Consequently we have the distributions of $T(j+1)$ and $T(j)$ as below:

\begin{center}
\begin{tabular}{ |c|c|c|c| }
	\hline
	$\ixjp$ & $\iyjp$ & $T(j+1)$ & $T(j)$  \\ 
	\hline
 0 & 0 & $1 + \gamma(p_0p_1)$ & $0$    \\  
	\hline
 1 & 1 & $1 + \gamma(1-p_0p_1)$ & $\gamma(p_0)\oplus\gamma(p_1)$    \\  
	\hline
 0 & 1 & $1 + \gamma(p_0)$ & $\gamma(p_0)$    \\  
	\hline
 1 & 0 & $1 + \gamma(p_1)$ & $\gamma(p_1)$    \\  
	\hline
\end{tabular}
\end{center}

Note that $1 - p_1 = p_0 \geq p_0p_1$ so $1-p_0p_1\geq p_1$. The conclusion follows from Lemma \ref{lem:bernoulliOrder}.
\end{proof}

\begin{lemma}
	\label{lem:CjSize1}
	If $\sizeof{\spotpotcands(j)} = 2$ and $\sizeof{\spotpotcands(j+1)} = 1$, then $T(j+1)\succ T(j)$.
\end{lemma}
\begin{proof}
	Still, $\spotpotcands(j+1)\cap\spotpotcands(j)=\varnothing$. So $Q^0_{j+1} = 0$ by equation \eqref{equ:Rd0DecomposeV2} while $O'_{j} = 0$ by equation \eqref{equ:Oj}. Meanwhile, the distributions of $Q^0_j$ and $O''_{j+1}$ can be obtained from Lemma \ref{lem:QjCj2} and Lemma \ref{lem:Oj1Cj_2}:

\begin{center}
\begin{tabular}{ |c|c|c|c| }
	\hline
	$\ixj$ & $\iyj$ & $O''_{j+1}$ & $Q^0_j$  \\ 
	\hline
 0 & 0 & $\gamma(p_0)\oplus\gamma(p_1)$ & $\gamma(p_0p_1)$    \\  
	\hline
 1 & 1 & $2$  & $\gamma(1-p_0p_1)$   \\  
	\hline
 0 & 1 & $\gamma(p_0)\oplus \gamma(p_0) \oplus \gamma(p_1)$  & $\gamma(p_0)$  \\  
	\hline
 1 & 0 & $\gamma(p_1)\oplus \gamma(p_1) \oplus \gamma(p_0)$  & $\gamma(p_1)$   \\  
	\hline
\end{tabular}
\end{center}
Obviously $O''_{j+1}\succ Q^0_j$ so $T(j+1)\succ T(j)$.
\end{proof}

\begin{lemma}
	\label{lem:CjCj1Size2NonOverlap}
	If $\sizeof{\spotpotcands(j)} = 2$, $\sizeof{\spotpotcands(j+1)} = 2$ and $\spotpotcands(j+1)\cap\spotpotcands(j)=\varnothing$. Then $T(j+1)\succ T(j)$.
\end{lemma}
\begin{proof}
	In this case, by Lemma \ref{lem:QjCj2} and Lemma \ref{lem:Oj1Cj_2}, we have the distributions of $Q^0_j$ and $O''_{j+1}$ as below
\begin{center}
\begin{tabular}{ |c|c|c|c| }
	\hline
	$\ixj$ & $\iyj$ & $O''_{j+1}$ & $Q^0_j$  \\ 
	\hline
 0 & 0 & $\gamma(p_0)\oplus\gamma(p_1)$ & $\gamma(p_0p_1)$   \\  
	\hline
 1 & 1 & $2$  & $\gamma(1 - p_0p_1)$ \\  
	\hline
 0 & 1 & $\gamma(p_0)\oplus \gamma(p_0) \oplus \gamma(p_1)$ & $\gamma(p_0)$  \\  
	\hline
 1 & 0 & $\gamma(p_1)\oplus \gamma(p_1) \oplus \gamma(p_0)$ & $\gamma(p_1)$  \\  
	\hline
\end{tabular}
\end{center}
in each case,  
\begin{equation}\label{equ:OppsuccQ0j}
	O''_{j+1}\succ Q^0_j
\end{equation}

Similarly, by Lemma \ref{lem:QjCj2} (substitute $j\rightarrow j+1$) and Lemma \ref{lem:OjCj1_2}, we have the distributions of $Q^0_{j+1}$ and $O'_{j}$ as below

\begin{center}
\begin{tabular}{ |c|c|c|c| }
	\hline
	$\ixjp$ & $\iyjp$& $Q^0_{j+1}$ & $O'_j$   \\ 
	\hline
 0 & 0 & $\gamma(p_0p_1)$ & $0$   \\  
	\hline
 1 & 1 & $\gamma(1 - p_0p_1)$ & $\gamma(p_0)\oplus\gamma(p_1)$   \\  
	\hline
 0 & 1 &  $\gamma(p_0)$ & $\gamma(p_0)$  \\  
	\hline
 1 & 0 & $\gamma(p_1)$ & $\gamma(p_1)$   \\  
	\hline
\end{tabular}
\end{center}

other than the $\ixjp = 1$ and $\iyjp = 1$ case, $Q^0_{j+1}\succ O'_j$. Combining with equation \eqref{equ:OppsuccQ0j}, we know that $T(j+1)\succ T(j)$ in all these cases. 

The only thing left is to prove for the $\ixjp = 1$ and $\iyjp = 1$ case. But in this scenario, we always have $Q^0_{j+1}\succ Q^0_{j}$ and $O''_{j+1}\succ O'_j$, and $T(j+1)\succ T(j)$ follows as well.
\end{proof}

\begin{lemma}
	\label{lem:CjCj1Size2OverlapSize1}
	If $\sizeof{\spotpotcands(j)} = 2$, $\sizeof{\spotpotcands(j+1)} = 2$ and $\sizeof{\spotpotcands(j+1)\cap\spotpotcands(j)}=1$. Then $T(j+1)\succ T(j)$.
\end{lemma}
\begin{proof}
	Without loss of generality, suppose $x_j = y_{j+1}$ and $x_{j+1}\neq y_j$. In this case $\spotpotcands(j+1)\setminus\spotpotcands(j) = \left\{ x_{j+1} \right\}$ and $\spotpotcands(j)\setminus\spotpotcands(j+1) = \left\{ y_{j} \right\}$. By equation \eqref{equ:Oj},
	\begin{equation*}
  \begin{split}
	O'_j & = \sum_{d\in\spotpotcands(j+1)\setminus\spotpotcands(j)}\indd{\ranker_*(d)<j} \\
	& = \indd{\ranker_*(x_{j+1})<j} \\
	& = \indd{\ranker_1(x_{j+1})<j}\indd{x_{j+1}\in\prodset_1}\\
	& = \indd{\ranker_1(x_{j+1})<j+1}\indd{x_{j+1}\in\prodset_1}.
  \end{split}
	\end{equation*}
	Meanwhile with equation \eqref{equ:Rd0DecomposeV2} (substitute $j\rightarrow j+1$),
\begin{equation*}
	\begin{split}
		Q^0_{j+1} & = \sum_{d\in\spotpotcands(j+1)}\indd{\ranker(d)<\ranker(x_{j+1})}\Big|_{x_{j+1}\in\prodset_0} \\
		& = \indd{\ranker(y_{j+1}<\ranker(x_{j+1}))}\Big|_{x_{j+1}\in\prodset_0}
	\end{split}
\end{equation*}

By equation \eqref{equ:expected_value_delta_j} (substitute $j\rightarrow j+1$), together with the fact that $x_j = y_{j+1}$ so $\ranker_0(y_{j+1}) = j < j+1$ and $\iyjp = 1$, the distribution of $Q^0_{j+1}$ and $O'_j$ can be summarized as in the following table:

\begin{center}
\begin{tabular}{ |c|c|c| }
	\hline
	$\ixjp$ & $Q^0_{j+1}$ & $O'_j$   \\ 
	\hline
 1  & $\gamma(1 - p_0p_1)$ & $\gamma(p_1)$   \\  
	\hline
 0  &  $\gamma(p_0)$ & 0  \\  
	\hline
\end{tabular}
\end{center}
Consequently $Q^0_{j+1}\succ O'_j$.

Similarly, by equation \eqref{equ:Ojplus1}
\begin{equation*}
  \begin{split}
	O''_{j+1} & = \sum_{d\in\spotpotcands(j)\setminus\spotpotcands(j+1)}\indd{\ranker_*(d)<j+1} \\
	& = \indd{\ranker_*(y_j)<j+1} \\
	& = \indd{y_j\in\prodset_1} + \indd{y_j\in\prodset_0}\indd{\ranker_0(y_j)<j+1} \\
	& = \indd{y_j\in\prodset_1} + \indd{y_j\in\prodset_0}\indd{\ranker_0(y_j)<j}.
  \end{split}
\end{equation*}
Meanwhile with equation \eqref{equ:Rd0DecomposeV2},
\begin{equation*}
	\begin{split}
		Q^0_{j} & = \sum_{d\in\spotpotcands(j)}\indd{\ranker(d)<\ranker(x_{j})}\Big|_{x_{j}\in\prodset_0} \\
		& = \indd{\ranker(y_{j}<\ranker(x_{j}))}\Big|_{x_{j}\in\prodset_0}.
	\end{split}
\end{equation*}

By equation \eqref{equ:expected_value_delta_j}, together with the fact that $x_j = y_{j+1}$ so $\ranker_1(x_{j}) = j+1 > j$ and $\ixj = 0$, the distribution of $Q^0_{j}$ and $O''_{j+1}$ can be summarized as in the following table:

\begin{center}
\begin{tabular}{ |c|c|c| }
	\hline
	$\iyj$ & $O''_{j+1}$ & $Q^0_j$  \\ 
	\hline
 0 & $\gamma(p_1)$ & $\gamma(p_0p_1)$   \\  
	\hline
 1 & $1$ & $\gamma(p_0)$  \\  
	\hline
\end{tabular}
\end{center}
Therefore $O''_{j+1}\succ Q^0_j$. Combining this with $Q^0_{j+1}\succ O'_j$ we get $T(j+1)\succ T(j)$.
\end{proof}

\begin{lemma}
	\label{lem:CjCj1Size2OverlapSize2}
	If $\sizeof{\spotpotcands(j)} = 2$, $\sizeof{\spotpotcands(j+1)} = 2$ and $\sizeof{\spotpotcands(j+1)\cap\spotpotcands(j)}=2$, then $T(j+1)\succ T(j)$.
\end{lemma}
\begin{proof}
	In this case $x_j = y_{j+1}$ and $y_j = x_{j+1}$. So $\spotpotcands(j+1)\setminus\spotpotcands(j) = \varnothing$ and $\spotpotcands(j)\setminus\spotpotcands(j+1) = \varnothing$. By equation \eqref{equ:Oj} and \eqref{equ:Ojplus1}, $O'_j = O''_{j+1} = 0$. Meanwhile by equation \eqref{equ:expected_value_delta_j}, $Q^0_j\disteq \gamma(p_0p_1)$ and $Q^0_{j+1}\disteq \gamma(1-p_0p_1)$ so $Q^0_{j+1}\succ Q^0_j$ and $T(j+1)\succ T(j)$.
\end{proof}

Combining Lemma \ref{lem:CjCj1Size1} to Lemma \ref{lem:CjCj1Size2OverlapSize2}, Theorem \ref{thm:ntreat1_unicorn_monotonic} is proved.

\newpage
\bibliographystyle{ims}
\bibliography{reference}


\end{document}